\documentclass[a4paper,11pt]{article}
\usepackage{jheppub} % for details on the use of the package, please see the JINST-author-manual
\usepackage{lineno}
\usepackage{physics}
\usepackage{mathtools}
\usepackage{tensor}
\usepackage{amsmath}
\usepackage{amsthm}
\usepackage{amssymb}
\usepackage{nicefrac}
\usepackage{bbold}
\usepackage{BOONDOX-cal} % more calligraphic letters, in particular H

\usepackage{caption}
\usepackage{subcaption}

\usepackage[nameinlink]{cleveref}

\usepackage[english]{babel}
\usepackage[utf8]{inputenc}
\usepackage[T1]{fontenc}
\usepackage[autostyle=true]{csquotes}
\usepackage[outline]{contour}
\usepackage{tikz}
\usepackage{tikz-3dplot}
\usepackage[compat=1.1.0]{tikz-feynman}
\usepackage{steinmetz}
\usetikzlibrary{calc,positioning,arrows,arrows.meta,decorations.markings}

\title{\boldmath Bell states for fermions in loop quantum gravity}

\author{Hanno Sahlmann}
\emailAdd{hanno.sahlmann@fau.de}

\author{and Martin Zei\ss}
\emailAdd{martin.zeiss@fau.de}

\affiliation{Institute for Quantum Gravity (IQG), Department of Physics,\\
Friedrich-Alexander-Universität Erlangen-Nürnberg (FAU),\\
Staudtstr. 7, 91058 Erlangen, Germany}

\abstract{Fermion fields are fundamental for the description of nature and also fit very naturally into the framework of loop quantum gravity. Motivated partially by proposals to use gravitationally mediated entanglement of matter as a witness for the quantum nature of gravity, we investigate how such entanglement can be defined and investigated in loop quantum gravity. In particular, we ask how a pair of fermions in a Bell state could be described in loop quantum, gravity. 

We demonstrate that the notion of fermionic entanglement in loop quantum gravity is subtle, by showing that some potential ways to define it fail. We then investigate a kinematical observable involving both, fermionic and gravitational degrees of freedom, the component of the fermion spin normal to a surface. We study its properties, and compare it to the standard operator for components of spin in a given direction in quantum mechanics. Using these normal components of spin, we define a kinematical observable that measures the correlation between space-like separated fermions which closely mirrors the CHSH observable. Finally, we exhibit states of the fermions coupled to quantum geometry that violate the Bell-CHSH inequality.}

\begin{document}
\maketitle
\flushbottom
\tikzset{->-/.style={decoration={
  markings,
  mark=at position #1 with {\arrow{>}}},postaction={decorate}}
}

\tikzfeynmanset{ with arrow/.style = {
   decoration={
     markings,
     mark=at position 0.5
          with {\arrow[xshift=2.5pt]{Stealth[width=4pt,length=5pt]}}
     },
   postaction=decorate}
}
%see 
%https://tex.stackexchange.com/questions/368271/change-edges-arrow-style-with-tikz-feynman

\NewDocumentCommand\tcoupling{mmmmO{}}{
\begin{tikzpicture}[baseline=($(j1.base)!.5!(j3.base)$)]
\begin{feynman}[inline=($(j1.base)!.5!(j3.base)$)]
    \vertex [small, dot] (j12) {}; 
    \vertex [above left=of j12] (j1) {
    \ifthenelse{\equal{#1}{0.5}}{\(\frac{1}{2}\)}{\(j_{#1}\)}
    }; 
    \vertex [above right=of j12] (j2) {
    \ifthenelse{\equal{#2}{0.5}}{\(\frac{1}{2}\)}{\(j_{#2}\)}
    }; 
    \vertex [below=of j12, small, dot] (j34) {};
    \vertex [below left=of j34] (j3) {
    \ifthenelse{\equal{#3}{0.5}}{\(\frac{1}{2}\)}{\(j_{#3}\)}
    }; 
    \vertex [below right=of j34] (j4) {
    \ifthenelse{\equal{#4}{0.5}}{\(\frac{1}{2}\)}{\(j_{#4}\)}
    }; 
    \ifthenelse{\equal{#1}{0.5}}{
    \diagram*{
        (j12) -- [charged scalar] (j1),
        (j12) -- [fermion] (j2),
        (j34) -- [fermion, edge label'=\({#5}\)] (j12),
        (j34) -- [fermion] (j3),
        (j34) -- [fermion] (j4)
    };
    }{
    \ifthenelse{\equal{#2}{0.5}}{
    \diagram*{
        (j12) -- [fermion] (j1),
        (j12) -- [charged scalar] (j2),
        (j34) -- [fermion, edge label'=\({#5}\)] (j12),
        (j34) -- [fermion] (j3),
        (j34) -- [fermion] (j4)
    };
    }{
    \ifthenelse{\equal{#3}{0.5}}{
    \diagram*{
        (j12) -- [fermion] (j1),
        (j12) -- [fermion] (j2),
        (j34) -- [fermion, edge label'=\({#5}\)] (j12),
        (j34) -- [charged scalar] (j3),
        (j34) -- [fermion] (j4)
    };
    }{
    \ifthenelse{\equal{#4}{0.5}}{
    \diagram*{
        (j12) -- [fermion] (j1),
        (j12) -- [fermion] (j2),
        (j34) -- [fermion, edge label'=\({#5}\)] (j12),
        (j34) -- [fermion] (j3),
        (j34) -- [charged scalar] (j4)
    };
    }{
    \diagram*{
        (j12) -- [fermion] (j1),
        (j12) -- [fermion] (j2),
        (j34) -- [fermion, edge label'=\({#5}\)] (j12),
        (j34) -- [fermion] (j3),
        (j34) -- [fermion] (j4)
    };
    }
    }
    }
    }
\end{feynman}
\end{tikzpicture}
}
\NewDocumentCommand\scoupling{mmmmO{}}{
\begin{tikzpicture}[baseline=($(j1.base)!.5!(j3.base)$)]
\begin{feynman}[inline=($(j1.base)!.5!(j3.base)$)]
    \vertex [small, dot] (j13) {}; 
    \vertex [right=of j13, small, dot] (j24) {};
    \vertex [above left=of j13] (j1) {
    \ifthenelse{\equal{#1}{0.5}}{\(\frac{1}{2}\)}{\(j_{#1}\)}
    }; 
    \vertex [above right=of j24] (j2) {
    \ifthenelse{\equal{#2}{0.5}}{\(\frac{1}{2}\)}{\(j_{#2}\)}
    }; 
    \vertex [below left=of j13] (j3) {
    \ifthenelse{\equal{#3}{0.5}}{\(\frac{1}{2}\)}{\(j_{#3}\)}
    }; 
    \vertex [below right=of j24] (j4) {
    \ifthenelse{\equal{#4}{0.5}}{\(\frac{1}{2}\)}{\(j_{#4}\)}
    }; 

    \ifthenelse{\equal{#1}{0.5}}{
    \diagram*{
        (j13) -- [charged scalar] (j1),
        (j24) -- [fermion] (j2),
        (j13) -- [fermion, edge label'=\({#5}\)] (j24),
        (j13) -- [fermion] (j3),
        (j24) -- [fermion] (j4)
    };
    }{
    \ifthenelse{\equal{#2}{0.5}}{
    \diagram*{
        (j13) -- [fermion] (j1),
        (j24) -- [charged scalar] (j2),
        (j13) -- [fermion, edge label'=\({#5}\)] (j24),
        (j13) -- [fermion] (j3),
        (j24) -- [fermion] (j4)
    };
    }{
    \ifthenelse{\equal{#3}{0.5}}{
    \diagram*{
        (j13) -- [fermion] (j1),
        (j24) -- [fermion] (j2),
        (j13) -- [fermion, edge label'=\({#5}\)] (j24),
        (j13) -- [charged scalar] (j3),
        (j24) -- [fermion] (j4)
    };
    }{
    \ifthenelse{\equal{#4}{0.5}}{
    \diagram*{
        (j13) -- [fermion] (j1),
        (j24) -- [fermion] (j2),
        (j13) -- [fermion, edge label'=\({#5}\)] (j24),
        (j13) -- [fermion] (j3),
        (j24) -- [charged scalar] (j4)
    };
    }{
    \diagram*{
        (j13) -- [fermion] (j1),
        (j24) -- [fermion] (j2),
        (j13) -- [fermion, edge label'=\({#5}\)] (j24),
        (j13) -- [fermion] (j3),
        (j24) -- [fermion] (j4)
    };
    }
    }
    }
    }
\end{feynman}
\end{tikzpicture}
}

\NewDocumentCommand{\ninejsymbol}{mmmmmmmmm}{ 
\begin{Bmatrix}
  #1 & #2 & #3 \\
  #4 & #5 & #6 \\
  #7 & #8 & #9
\end{Bmatrix}
}
\NewDocumentCommand{\sixjsymbol}{mmmmmm}{ 
\begin{Bmatrix}
  #1 & #2 & #3 \\
  #4 & #5 & #6 
\end{Bmatrix}
}

\newtheorem{prop}{Proposition}[section]
\newtheorem{definition}{Definition}[section]

% Hanno's commands

\newcommand{\expec}[1]{\langle #1\rangle}
\newcommand{\scpr}[2]{\langle#1\, \vert \, #2 \rangle}
\newcommand{\sscpr}[3]{\langle#1\, \vert \, #2 \, \vert \, #3\rangle} 
\newcommand{\mea}[1]{\text{d}#1\,}
\newcommand{\pro}[1]{\ket{#1}\bra{#1}}
\newcommand{\one}{\mathbb{1}}

% fermionic annihilation and creation operators 
\newcommand{\cre}[2]{c^\dagger{}^{#2}(#1)}
\newcommand{\ann}[2]{c_{#2}(#1)}

% Fock vacuum 
\newcommand{\vac}{\left\lvert 0 \right\rangle}

% spin operators 
\newcommand{\SprojE}[0]{S_{\mathcal{S}}}
\newcommand{\SprojF}[0]{2 S_{E'} \cdot \Sx{e_{n+1}(0)}{}}
\newcommand{\Sx}[2]{S^{#2}({#1})}
\newcommand{\Sxl}[2]{S_{#2}({#1})}
\newcommand{\Se}[2]{S_{#1}^{#2}}
\newcommand{\SE}[2]{S_{#1}^{#2}}

% su(2) rep
\newcommand{\jrep}[2]{\pi_{#1}({#2})}
%\newcommand{\jrep}[2]{\overset{#1}{#2}}

% various QT related

%\newcommand{\ket}[1]{\lvert\, #1\,\rangle}
%\newcommand{\bra}[1]{\langle\, #1\,\rvert}
%\newcommand{\pro}[1]{\ket{#1}\bra{#1}}
\newcommand{\gbra}[1]{(\,#1\,\rvert}
\newcommand{\gscpr}[2]{(\,#1\, \vert \, #2\, \rangle}

\renewcommand{\acomm}[2]{\left[\,#1\, ,\, #2\, \right]_+}
\section{Introduction}

Fermion fields are fundamental for the description of nature. Fermions also fit very naturally into the framework of loop quantum gravity (LQG). They transform nontrivially under diffeomorphisms and, in the variables used in loop quantum gravity, under the same gauge transformations as the gravitational fields. It was observed early that they also fit neatly into the quantum geometry framework of LQG \cite{Morales-Tecotl:1994rwu,Morales-Tecotl:1995oxb,Smolin:1994uz} and are strongly entangled with the quantum geometry already at the kinematical level. Therefore, they potentially constitute sensitive probes for observable effects from loop quantum gravity. 

In the present work, we establish three results. 
\begin{enumerate}
    \item We demonstrate that the notion of fermionic entanglement in loop quantum gravity is subtle, by showing that some obvious ways to investigate it fail. 
    \item We expand on an idea from \cite{Mansuroglu:2020acg} and investigate a kinematical observable, the component of the fermion spin normal to a surface, study its properties, and compare it to the standard operator for components of spin in quantum theory.
    \item Using the signs of spin components relative to surfaces, we define a kinematical observable that measures the correlation between space-like separated spins that closely mirrors the CHSH observable \cite{Clauser:1969ny}. We then exhibit states of the fermions coupled to quantum geometry that violate the Bell-CHSH inequality.
\end{enumerate}
Let us motivate this work and put it into context in the following. 

Fermions couple very naturally to the spin network excitations of loop quantum gravity. \cite{Morales-Tecotl:1994rwu,Thiemann:1997rq,Baez:1997bw,Bojowald:2007nu,Mansuroglu:2020acg,Lewandowski:2021bkt}. Since they transform under gauge transformations, fermionic excitations have to be combined with gravitational excitations already on the kinematical level. A typical state 
\begin{equation}
   \ldots{h}[A]_C{}^B \; {c}^\dagger_B(p)\; \ket{0}
\end{equation}
combines gravitational holonomies and fermionic excitations in such a way that they are maximally entangled. In this way, fermionic excitations in loop quantum gravity exhibit interesting properties already at the kinematical level. If it could be shown that they behave observably different due to some of these properties, this would constitute a sensitive probe of the framework of loop quantum gravity as a whole. Therefore, it is useful to compare the behavior of fermions in loop quantum gravity to those in quantum mechanics or standard quantum field theory, and the current work contributes to this effort. 

As a second motivation, and with an eye towards the dynamics of the theory, we point to very interesting proposals to establish the quantum nature of gravity by experimentally observing gravitational generation of entanglement \cite{Bose:2017nin,Marletto:2017kzi}. Going beyond merely establishing entanglement generation experimentally, it may one day be possible to make detailed measurements of this process, and thereby obtain clues as to the underlying quantum gravity theory. Therefore it is relevant to study gravitational entanglement generation in loop quantum gravity. A first step is to understand and distinguish entangled and non-entangled multi-fermion states in loop quantum gravity. The present work takes this first step. We identify states of two fermions coupled to gravity in which the spin components normal to surfaces show non-classical correlations and hence entanglement. We will report elsewhere \cite{Sahlmann_Zeiss_TBD} on using the kinematical states considered in the present work to establish entanglement generation in spin foam amplitudes.

Fermions were first introduced to loop quantum gravity in \cite{Morales-Tecotl:1994rwu,Morales-Tecotl:1995oxb}. Important refinements to the quantum theoretical description were achieved in \cite{Thiemann:1997rq}. In previous work \cite{Mansuroglu:2020acg}, one of us has explored fermion spin, and fermion spin addition in loop quantum gravity. It turns out that if the Ashterkar-Barbero connection \cite{Ashtekar:1986yd,BarberoG:1994eia} is used to parallel transport fermion spins, they have the same commutator relations and behavior under addition as would be expected from quantum field theory or quantum mechanics in a flat, classical spacetime. 

One problem left open in \cite{Mansuroglu:2020acg} was the calculation of spin-spin correlations. Expectation values 
\begin{equation}
    \expec{S^a(p)S^b(q)}_\Psi, \qquad a,b = 1,2,3 
\end{equation}
of the components  of spins at points $p,q$ in a spatial slice $\Sigma$, taken in a combined state of matter and gravity, are not well defined in loop quantum gravity due to the coordinate dependence of the spin components. Correlation functions 
\begin{equation}
    \expec{S^i(p)S^k(q)}_\Psi\qquad i,k = 1,2,3 
\end{equation}
for spin components in internal space ($i$ and $k$ label a basis of \(\mathfrak{su}(2)\) in this case) turn out to be trivial for $p\neq q$ for a large class of states. This is not as problematic as it might seem, due to the fact that they are gauge dependent and thus far from observable. But it raises the question of how fermion entanglement can be described and measured in loop quantum gravity. In the present work, we give one answer to this, by introducing operators measuring the spin components relative to the normal of surfaces. These operator are gauge invariant and covariant under spatial diffeomorphisms. Correlations defined using these operators \emph{can} be non-trivial. In fact, we exhibit Bell-like states for the fermions coupled to gravitational excitations. These states turn out not to be spin network states, but linear combinations thereof, such that the intertwiners at the positions of the fermions are entangled in a specific way. 
They thus carry intertwiner entanglement in the sense of \cite{Livine:2017fgq,Chen:2022rty,Bianchi_2023}. Intertwiner entanglement has also been important in the construction of other states with special properties, see for example \cite{Freidel_2010,Bianchi:2018fmq,Baytas:2018wjd}. More generally, long-range entanglement is a hallmark of quantum field theoretical ground states for matter, see for example \cite{Agullo:2023fnp,Agullo:2024har}. 

The structure of the article is as follows. In section \ref{sec:fermions_lqg} we review the formalism for coupling a fermion field to quantum geometry and various operators involving the spin of the fermions. We then demonstrate that certain obvious ways of demonstrating non-local correlations between fermions fail.  

Section \ref{sec:spin_surface_normal} contains definition and properties of the new operators. In section \ref{sec:bell} we construct a CHSH-like observable and calcluate its expectation values for various states, demonstrating non-classical correlations for fermion spins in loop quantum gravity. 

We close with a summary and discussion in section \ref{sec:summary}.

\section{Fermion spins in loop quantum gravity}
\label{sec:fermions_lqg}

Fermions were first introduced to loop quantum gravity in \cite{Morales-Tecotl:1994rwu,Morales-Tecotl:1995oxb}. Important refinements to the quantum theoretical description were achieved in \cite{Thiemann:1997rq,Thiemann:1997rt}. An alternative that differs in important aspects is \cite{Baez:1997bw}, but the present work is based on the form \cite{Thiemann:1997rq}, as it holds important advantages for the dynamics and is much more natural in the loop quantum gravity framework. A spin foam formalism having the the kinematical states of \cite{Thiemann:1997rq} as boundary states was given in \cite{Bianchi:2010bn,Han:2011as}. Many aspects of the classical canonical theory in the presence of fermions were clarified in \cite{Bojowald:2007nu}. 

In previous work \cite{Mansuroglu:2020acg,Mansuroglu:2020dga}, one of us has explored fermion spin, and fermion spin addition in loop quantum gravity. We will only summarize the setup here, and refer to the literature cited above for details. 

\subsection{From the classical to the quantum theory}\label{sec:classical_to_quantum}
On the classical level, a theory of spin $\frac{1}{2}$ fermions coupled to gravity is described by the Einstein-Cartan-Holst action \cite{Thiemann:1997rt,Mercuri:2006um,Bojowald:2007nu} together with the covariant version of the Dirac action 
\begin{align}
    S[e,\omega,\Psi] = \frac{1}{16 \pi G} &\int_M \text{d}^4x\; |\det e| \, e^{\mu}_Ie^\nu_J P^{IJ}{}_{KL} F^{KL}_{\mu\nu}(\omega) \nonumber \\
    +\frac{\mathrm{i}}{2} &\int_M \text{d}^4x\; |\det e|\left[\overline{\Psi} \gamma^Ie_I^\mu \nabla_\mu \Psi - c.c.\right].
    \label{eq:action}
\end{align}
A foliation into three-dimensional hypersurfaces $M=\mathbb{R}\times \Sigma$ is performed and the canonical analysis yields the canonical gravitational variables \cite{Thiemann:1997rt, Bojowald:2007nu}
\begin{align}
    \mathcal{A}_a^i &=\Gamma_a^i +\beta K_k^i + 2\pi G\beta \, \epsilon^i{}_{kl} e^k_a J^l \label{eq:connection} \\ 
    E^a_i &=\frac{1}{2}\epsilon_{ijk}\epsilon^{abc} e^i_ae^j_b .
\end{align}
Here 
\begin{equation}
    \Gamma^i= (\star\,\omega)^{iJ}\,n_J, \quad
    K^i=\omega^{iJ}\,n_J
\end{equation}
are components of $\omega$ obtained by decomposing it into electric and magnetic parts relative to a time-like internal vector field $n$. On-shell, they become spin connection and extrinsic curvature. The $J^l$ are the spatial components of the fermion current
\begin{equation}
    J^l=\overline{\Psi}\gamma^l\Psi .
\end{equation}
In the chiral basis, the Dirac fermion $\Psi$ and its conjugate momentum splits into its chiral components and one defines  \cite{Thiemann:1997rt,Thiemann:1997rq}   
\begin{align}
    \theta^A(x) &=\sqrt[4]{\det q} \;\Psi_\text{R}(x), \qquad \pi_\theta(x)=-i\theta^\dagger(x)
\end{align}
and similarly for the left-handed component. 
In the following, we will focus on  $(\theta, \pi_\theta)$ only. The non-vanishing anti-Poisson relations are  
\begin{equation}
\label{eq:antipoi}
    \left\{\theta^A(x), \pi_{\theta B}(y)\right\}_+= \delta^{A}_{B} \, \delta^{(3)}(x,y). 
\end{equation}
The action (\ref{eq:action}) yields, among other things, a contribution to the Gauss constraint 
\begin{equation}
\label{eq:gauss}
    G^\theta_i(x)=
    \theta^\dagger(x)\sigma_i\theta(x), 
\end{equation} so the fermions transform under the local SU(2) gauge transformations in the defining representation.  

In the canonical formulation of the quantum theory, a Weyl field is described  using fermionic creation and annihilation operators
\begin{equation}
\label{eq:Frep}
    \theta^A(x) = \ann{x}{A},  \qquad -i \pi_{\theta B}(y)= \cre{y}{A}
\end{equation}
which satisfy canonical anticommutation relations\footnote{Note that the operators $\theta, \theta^\dagger$ are tangent space scalars, although they come from half density fermion fields. More details can be found in \cite{Thiemann:1997rq,Thiemann_2007,Mansuroglu:2020acg}.}  
\begin{align}
    \acomm{\ann{x}{A}}{\cre{y}{B}} &= \delta_{x,y}\,\delta^A_B, \\
    \acomm{\cre{x}{A}}{\cre{y}{B}}=0, \qquad &\acomm{\ann{x}{A}}{\ann{y}{B}} = 0.
\end{align}
The action of these operator spans the fermionic Fock space $\mathcal{F}_-(\mathcal{h})$ over a one particle space $\mathcal{h}$ given by 
\begin{gather}
    \mathcal{h} = \{f:\Sigma \longrightarrow\mathbb{C}^2:\; f(x)\neq 0 \text{ only for finitely many } x \} , \nonumber \\
    \scpr{f}{f'}_\mathcal{h}=\sum_{x\in\Sigma}\overline{f(x)} f'(x).
\end{gather}
The gravitational observables act on cylindrical functions which form the Ashtekar-Lewandowski Hilbert space
\begin{equation}
\mathcal{H}_\text{AL}=L^2(\Sigma,\text{d}\mu_{\text{AL}})
\end{equation}
via multiplication of holonomies, and derivations $X_S$, respectively
\begin{align}
\label{eq:Grep}
    \jrep{j}{h}_e \Psi[A] &= \jrep{j}{h}_e[A] \Psi[A], \\
    \int_S E_i \Psi[A] &= \mathrm{i} (X_S^i\Psi)[A].
\end{align}
The derivations $\mathrm{i} X_S^i$ can be further described in terms of spin-like operators $J$ \cite{Thiemann_2007}. We will give some details below, in \eqref{eq:P_Ff_quantum}, \eqref{eq:js}.
The combined system of matter fields and gravitational degrees of freedom is given by the tensor product
\begin{equation}
\mathcal{H}=\mathcal{H}_\text{AL}\otimes \mathcal{F}_-(\mathcal{h}),
\end{equation}
so we extend \eqref{eq:Grep}, \eqref{eq:Frep} in the obvious way.
The fermion field transforms under the gauge transformations of the gravitational sector as
\begin{align}
\label{eq:gauge_trafos_fermions}
   % U_g\jrep{j}{h_e} U^{-1}_g &= g(s(e)) \cdot\jrep{j}{h_e} \cdot g^{-1}(t(e)),\\
    U_g \theta(x) U^{-1}_g &= g(x) \cdot \theta(x), \\
    U_g \pi_{\theta} (x) U^{-1}_g &= \pi_{\theta(x)}\cdot g^{-1}(x).
\end{align}
Angular momentum is not a covariant concept, but relative to the decomposition of (local) Lorentz transformations into rotations and boosts. In QFT on flat space, upon fixing a reference system by specifying a timelike normal $n$, one can recover the spin operator as the component of the generator of rotations, on spinor space. 

In \cite{Mansuroglu:2020acg} it was argued that the natural generalization of the spin density to curved spacetime is the generator of rotations, as defined by the internal version of the timelike normal vector $n$, on the Weyl spinor $\theta(x)$, 
\begin{equation}
\label{eq:internal_spin}
    \Sx{x}{i}= \frac{1}{2} \theta(x) \sigma^i \theta^\dagger(x). 
\end{equation}
Note that this generator is a tangent space scalar in the quantum theory, unlike the flat space generator, which is a density. This difference traces back to the special properties of the quantum theory and canonical loop quantum gravity \cite{Thiemann:1997rq}. 

In flat space, spin is a vector, and in the present context the local tangent space spin could in principle be defined as 
\begin{equation}
    \Sx{x}{a}= \frac{1}{2} \theta(x) \sigma^i e_i^a(x)\theta^\dagger(x).  
\end{equation}
However the triad $e_i^a(x)$ is not realized as a well defined operator on 
$\mathcal{H}_\text{AL}$, so while $ \Sx{x}{i}$ is a well defined operator, $ \Sx{x}{a}$ is not. Fortunately, there are interesting functionals of $\Sx{x}{a}$ that can be written in terms of $\Sx{x}{i}$ and gravitational field components that correspond to well defined operators. A trivial example is the modulus squared of the local spin, 
\begin{equation}
S^2 (x) = \Sx{x}{a}\Sx{x}{b} q_{ab}(x) = \Sx{x}{i}\Sx{x}{j} \delta_{ij},
\end{equation}
where $\delta$ is to be understood as the internal spatial metric. A more relevant example is the (area weighted) normal component of spin, relative to a surface $\mathcal{F}$, 
\begin{equation}
    S\cdot n{}_\mathcal{F}=\int_{\mathcal{F}}\Sx{x}{a}\,  q_{ab} (x)\, n^b(x)\,\dd A, 
\end{equation}
Where $q$ is the spatial metric, $n$ the unit normal to $\mathcal{F}$, and $\dd A$ its area element induced by $q$. We will consider this quantity in detail in section \ref{sec:spin_surface_normal}.
\subsection{Total spin and problems with spin correlations}
As a simple consequence of \eqref{eq:gauge_trafos_fermions}, \eqref{eq:internal_spin}, one has 
\begin{equation}
     U_g \Sx{x}{i} U^{-1}_g = \left(\pi_1(g(x))S(x)\right)^{i}
\end{equation}
and 
\begin{equation}
    \comm{\Sx{x}{i}}{\Sx{x}{j}}=\mathrm{i}\,\epsilon^{ij}{}_k\,\Sx{x}{k}.
\end{equation}

Since the geometry is curved, spins have to be parallel transported to be added. In the quantum theory, this can be done for example by using using the Ashtekar-Barbero connection. In \cite{Mansuroglu:2020acg} the parallel transport of spin is defined as
\begin{equation}
    \Se{e}{i}= \pi_{j=1}(h_e)^{i}{}_j\,\Sx{s(e)}{j} 
\end{equation}
where $s(e)$ is the starting point of $e$. Importantly, the parallel transported spin is again a quantum theoretical spin:
\begin{equation}
    \comm{\Se{e}{i}}{\Se{e}{j}}=\mathrm{i}\,\epsilon^{ij}{}_k\,\Se{e}{k}.
\end{equation}
Under a gauge transformations it transforms at the final point of $e$. Collecting several spins along edges $E=\{e_I\}$ that end at the same point $p$ yields the total spin 
\begin{equation}
    S_E^i:= \sum_I \Se{e_I}{i}.
\end{equation}
Some simple eigenstates of $S_E^2$ have been discussed in \cite{Mansuroglu:2020acg}. As an example, consider two-fermion state 
\begin{equation}
\label{eq:singlet}
    \Psi_2 = \big(h_{e_2}\theta^\dagger(p_2)\big)_{A}\;   \epsilon^{AB}\; \big(h_{e_1}\theta^\dagger(p_1)\big)_{B}\, \ket{0} 
\end{equation}
in which two fermionic excitations are connected by a holonomy to form a gauge invariant state, see figure \ref{fig:two_fermions}. As one might guess, this state turns out  to be an eigenstate of the total spin $S_{E}^2$ for $E={e_1,e_2}$, with the eigenvalue 0 \cite{Mansuroglu:2020acg}.
\begin{figure}
    \centering
    \includegraphics[width=\linewidth]{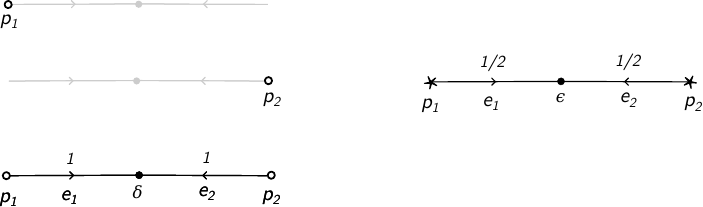}
    \caption{Graphical representation of the two-fermion state $\Psi_2$ (right) as well as spin operators $S^{k_i}(p_i)$ (upper left) and the total spin (lower left).}
    \label{fig:two_fermions}
\end{figure}
More generally, it can be shown that eigenstates to $S_E^2$ for any $E$ can be constructed by creating fermions at the beginning of (some of) the edges, and parallel transporting them, along the edges, to their common endpoint $p$. At $p$ they have to be coupled using particular intertwiners obtained by suitable action of the symmetric group \cite{Grosse_BSc}, 
\begin{equation}
    \Psi_n = \iota_{A_1 A_2 \ldots A_{|E|}}\prod_i \left(h_{e_i}  \theta^\dagger(s(e_i))\right)^{A_I} \; \ket{0}.  
\end{equation}
In this way one can show that the operator $S_E^2$ has the spectrum expected from quantum mechanics in flat space \cite{Grosse_BSc}. 

It turns out to be more difficult to define observables that measure non-trivial spin-spin correlations. The most obvious candidate, the $n$-point functions 
\begin{equation}
    W^{a_1\ldots a_n}(x_1,\ldots,x_n)
    :=\expec{\Sx{x_1}{a_1}\, \Sx{x_2}{a_2}\,\ldots\, \Sx{x_n}{a_n}}_{\Psi}
\end{equation}
in tangent space, for a state $\Psi\in \mathcal{H}$ are not accessible, because the triads $e^a_i(x)$ and hence the $\Sx{x}{a}$ are not well defined operators on $\mathcal{H}$.

The next best candidate, the $n$-point functions  
\begin{equation}
    W^{k_1\ldots k_n}(x_1,\ldots,x_n)
    :=\expec{\Sx{x_1}{k_1}\, \Sx{x_2}{k_2}\,\ldots\, \Sx{x_n}{k_n}}_{\Psi}
\end{equation}
in internal space \emph{are} well defined. However, for gauge invariant, and hence physically relevant, states $\Psi$, they show no non-local correlations whatsoever. 
\begin{prop}
    Let $n>0$, and $x_1,\ldots,x_n$ pairwise different points. Let $\Psi$ be gauge invariant, i.e., $U_g \Psi=\Psi$ for any gauge transformation $g$. Then 
\begin{equation}
    W^{k_1\ldots k_n}(x_1,\ldots,x_n)=0 \quad \text{ for all } \quad k_1,\ldots, k_n \in \{1,2,3\}^{n}.
\end{equation}    
\end{prop}
\begin{proof}
Let $X\coloneqq\{x_1,\ldots,x_n\}$. Consider the states 
\begin{equation}
    \Psi_{n}^{k_1\ldots k_n} := \Sx{x_1}{k_1}\, \Sx{x_2}{k_2}\,\ldots\, \Sx{x_n}{k_n} \Psi. 
\end{equation}
Because of the invariance of $\Psi$, the states form a multiplet under gauge transformations, 
\begin{equation}
\label{eq:multiplett}
   U_g\, \Psi_{n}^{k_1\ldots k_n} = \pi_1(g(x_1))\indices{^{k_1}_{m_1}} \ldots \pi_1(g(x_n))\indices{^{k_n}_{m_n}}\, \Psi_{n}^{m_1\ldots m_n}.  
\end{equation}
The set of gauge transformations $\mathcal{G}_{X}=\{g | g(x)=\one \text{ for all } x \in \Sigma / X \}$ that act trivially outside of $X$ forms a group isomorphic to SU(2)$^n$. This group is unitarily represented on $\mathcal{H}$. Since $\mathcal{G}_{X}$ is compact and unitary, $\mathcal{H}$ decomposes into orthogonal irreducible subspaces. 

Equation \eqref{eq:multiplett} shows that the subspace spanned by the $\Psi_{n}^{k_1\ldots k_n}$ is an invariant subspace, and a standard result in algebra shows that, because the action is a direct product of irreducibles of SU(2), it is irreducible for $\mathcal{G}_{X}\simeq \text{SU(2)}^n$. As such, it is orthogonal to all other irreducible subspaces of $\mathcal{H}$, in particular of the span of $\Psi$, on which $\mathcal{G}_{X}$ acts in the trivial representation. 
\end{proof}
This proposition shows, for example, that if the two-fermion singlet state $\Psi_2$ from \eqref{eq:singlet} were to have non-local correlations, they are not detected by $W_2$. Upon reflection, this result is not surprising, however. One can show that in a precise sense, the excited gravitational degrees of freedom created by the holonomies in $\Psi_2$ are maximally entangled to the fermionic excitation at their respective endpoints. Monogamy of entanglement \cite{Coffman:1999jd} then dictates that the fermions can not be entangled with each other in $\Psi_2$. 

There is another way that one might hope to find the non-local correlations between fermionic degrees of freedom, but it fails for closely related reasons. Consider the gravitational trace 
\begin{equation}
    \tr_G (\,\cdot\,) \coloneqq \sum_{S\in \mathcal{S}} \sscpr{S}{\,\cdot\,}{S}  
\end{equation}
where the sum is over the orthonormal basis $\mathcal{S}$ of generalized (i.e., non-gauge invariant) spin networks \cite{Ashtekar:1996eg}. One can show that this is a well defined map from pure density matrices $\pro{\psi}$ over $\mathcal{H}$ to density matrices over $\mathcal{F}_-$ that intertwines the action of the gauge transformations on $\mathcal{H}$ and $\mathcal{F}_-$ \cite{Kossmann_BSc},
\begin{equation}
     \tr_G (U_g \,\cdot\,U^\dagger_g) =  U_g^{\mathcal{F}_-} \tr_G (\,\cdot\,) (U^{\mathcal{F}_-}_g)^\dagger. 
\end{equation}
One can now ask what density matrices over $\mathcal{F}_-$ one can obtain as an image of $ \tr_G$. 
\begin{prop}
\label{pr:trace}
    Let $\Psi_n\in \mathcal{H}$ be an $n$-fermion state that is gauge invariant, i.e., $U_g \Psi=\Psi$ for any gauge transformation $g$. Then 
    \begin{equation}
     \rho^{\mathcal{F}_-} \coloneqq \tr_G (\pro{\Psi})
    \end{equation}
    is maximally mixed on the subspace spanned by the $n$ fermions, i.e., 
\begin{equation}
    \rho^{\mathcal{F}_-} = C \sum_{A_1,\ldots, A_n} \pro{\Phi_{A_1 \ldots A_n}}, \qquad  \Phi_{A_1 \ldots A_n} = \prod_{i}  \theta_{A_i}^\dagger(x_i) \ket{0}_{\mathcal{F}_-}
\end{equation}
for some constant $C$. 
\begin{proof}
    Since $\Psi$ is gauge invariant, and $\tr_G$ intertwines the action of gauge transformations on $\mathcal{H}$ and $\mathcal{F}_-$, $\rho^{\mathcal{F}_-}$ must be gauge invariant itself. $\tr_G$ does not change number or position of the fermionic excitations. Since the action of gauge trafos is local, there is, up to constant factors, only one density matrix containing $n$ fermions at the correct positions that is gauge invariant.  
\end{proof}    
\end{prop}
Note that Prop. \ref{pr:trace} can also be read as a confirmation that fermions and gravitational degrees of freedom are maximally entangled wit each other: The entanglement entropy is maximal, since the result of the gravitational trace is maximally mixed. 

The two results above show that any notion of entanglement between fermions in gauge invariant states in loop quantum gravity must be somewhat subtle. It can not be detected straightforwardly in $n$-point functions or in reduced density matrices on the fermionic Hilbert space. There is, thus, the question as to how non-local entanglement of fermion spins, demonstrated so beautifully in various experiments, can be described in loop quantum gravity. The rest of this work is devoted to giving one possible answer. 

\section{Spins relative to surface normals} \label{sec:spin_surface_normal}
Carrying over the notion of the projection of a fermion spin \(\vec{S}\) into a certain direction \(\vec{n}\in\mathbb{R}^3\) from quantum mechanics to loop quantum gravity is not straightforward. The reason is that it is not known how to define operators transforming as tangent space vectors, such as $S^a(x)$, in LQG. Fixing an internal direction on the other hand runs counter to implementing the Gau\ss{} constraint which imposes \(SU(2)\) invariance at the vertices. As pointed out in \cite{Mansuroglu:2020acg}, a way out is to project the spin onto a reference field given by the normal of a surface. 
\subsection{Construction of the cosine operator}
As a classical starting point, consider the definition of the flux of a test function \(f^i:\Sigma\rightarrow\mathfrak{su}(2)\) (with compact support) through the surface \(\mathcal{F}\subseteq\Sigma\) \cite{Haggard:2023tnj}
\begin{equation}\label{eq:P_Ff}
    P_{\mathcal{F},f}(E)=\int_{\mathcal{F}} f^iE^a_i\varepsilon_{abc}\,\dd x^b\dd x^c .
\end{equation}
Here \(E^a_i=\sqrt{\abs{\operatorname{det}(q)}}\,e^a_i\) is the densitized triad on \(\Sigma\). %
We parametrize \(\mathcal{F}\) by \((s,t)\in\mathbb{R}^2\) which allows us to define the surface normal as
\begin{equation}
    n^a\coloneqq\frac{\tilde{n}^a}{\sqrt{\tilde{n}^2}}\,,\quad\tilde{n}^{a}\coloneqq q^{aa^\prime}\varepsilon_{a^\prime bc}\frac{\partial x^b}{\partial s}\frac{\partial x^c}{\partial t} .
\end{equation}
A few lines of calculation show that
\begin{equation}
    \tilde{n}^2=\frac{\operatorname{det}(q^{(2)})}{\operatorname{det}(q)}
\end{equation}
where \(q^{(2)}=\iota^*q\) is the metric on \(\mathcal{F}\) induced by the pullback of \(q\) via the inclusion \(\iota:\mathcal{F}\rightarrow\Sigma\). Hence
\begin{equation}\label{eq:integral f_q_n}
    P_{\mathcal{F},f}(E)=\int_{\mathcal{F}}f^aq_{ab}n^b\,\dd A
\end{equation}
with \(f^a=f^ie^a_i\) and \(\dd A\coloneqq\sqrt{\operatorname{det}(q^{(2)})}\,\dd s\,\dd t\).
\cref{eq:integral f_q_n} is nothing but the integrated scalar product on tangent space between the surface normal \(n^a\) and \(f^a\). Quantizing the flux operator following the standard procedure as e.g. done in \cite{Thiemann_2007} one obtains
\begin{equation}\label{eq:P_Ff_quantum}
    P_{\mathcal{F},f}[\mathcal{F}]=\frac{\kappa\gamma}{2}\sum_{x\in\mathcal{F}} f^i(x) \qty(J^{x,S\uparrow}_i-J^{x,S\downarrow}_i)
\end{equation}
%We can now quantize the classical expression for the spin flux, which is just \cref{eq:P_Ff} with \(f^i\) replaced by \(\tilde{S}^i\), following the standard procedure as e.g. done in \cite{Thiemann_2007}. We have
%\begin{equation}\label{eq:SP_definition}
%    \widehat{SP}_{\mathcal{F},\epsilon}=\frac{\kappa\gamma}{2}\sum_{x\in\mathcal{F}} \qty(\sum_{y\in B_\epsilon(x)} \hat{S}^i(y)) \qty(\hat{J}^{x,S\uparrow}_i-\hat{J}^{x,S\downarrow}_i)
%\end{equation}
where the quantum spin operators associated to the equivalence class \([p]\) of paths  beginning at \(x\) and sharing initial segments are defined by
\begin{equation}
\label{eq:js}
    J^{x,\mathcal{F}\uparrow}_i\coloneqq\sum_{\text{\([p]\) going up}} J^{x,p}_i\,, \quad J^{x,\mathcal{F}\downarrow}_i\coloneqq\sum_{\text{\([p]\) going down}} J^{x,p}_i \,.
\end{equation}
%The sums in \ref{eq:SP_definition} range over all points in the surface \(\mathcal{F}\) and the volume \(B_\epsilon(x)\) respectively. As in the case of the usual flux operator, these sums become finite when the operator is applied to a spin network function based on a graph \(\gamma\). Namely \(x\) (resp. \(y\)) ranges over all vertices of \(\gamma\) which intersect \(\mathcal{F}\) (resp. \(B_\epsilon(x)\)). When shrinking the volume to \(x\) we obtain the scalar product operator as
A priori, the sum in \cref{eq:P_Ff_quantum} ranges over all points in the surface. However when acting on a spin network function based on a graph \(\gamma\) only the intersections of \(\mathcal{F}\) with the vertices or edges of \(\gamma\) will yield a non-zero contribution, thus rendering the sum finite.
%\begin{equation}
%    \hat{S}[\mathcal{F}]\coloneqq\lim_{\epsilon\rightarrow0}\widehat{SP}_{\mathcal{F},\epsilon}=\frac{\kappa\gamma}{2}\sum_{x\in\mathcal{F}} \hat{S}^i(x)\qty(\hat{J}^{x,S\uparrow}_i-\hat{J}^{x,S\downarrow}_i) \,.
%\end{equation}

In order to arrive at the spin projection we simply specify \(f^i\) to be the spin operator \(S^i\) as defined in \cref{eq:internal_spin}. We arrive at 
\begin{equation}
    {S\cdot n}_{\mathcal{F}}=\frac{\kappa\gamma}{2}\sum_{x\in\mathcal{F}} S^i(x)\qty(J^{x,\mathcal{F}\uparrow}_i-J^{x,\mathcal{F}\downarrow}_i).
\end{equation}
Due to the definition of \(S^i\), only those vertices contribute where at least one fermion sits. %In the case of exactly one fermion at \(x\) the operator acts as \(\frac{\hbar}{2}\sigma_i\).%
%\footnote{As shown in \cite{Mansuroglu:2020acg} we have \(\hat{S}^i(x)\ket{0}=0\) and \(\qty[\hat{S}^i(x),{c^\dagger}^A(y)]=\frac{\hbar}{2}\,\delta_{x,y}{\sigma_i}\indices{^A_B}{c^\dagger}^B(y)\) from which the statements follow.}

Note that there is no ordering ambiguity as the spin operator commutes with the \(J_i\). 
%This is due to the fact that both operators act on different Hilbert spaces in the product \(\mathcal{H}_{\text{AL}}\otimes\mathcal{F}_-(h)\).

As \(S^i(x)\) annihilates the zero fermion state and two fermion singlet states, in what follows we will always consider only one fermion at each vertex, i.e. the fermion state is proportional to \(c^\dagger\ket{0}\). In this case, as shown in \cite{Mansuroglu:2020acg}, the action of the spin operator reduces to multiplication with the quantum mechanical spin operator \(S^i_{\text{QM}}=\frac{\hbar}{2}\sigma_i\)%
\footnote{This statement follows from \(S^i(x)\ket{0}=0\) and \(\qty[S^i(x),{c^\dagger}^A(y)]=\frac{\hbar}{2}\,\delta_{x,y}{\sigma_i}\indices{^A_B}{c^\dagger}^B(y)\).}%
. By abuse of notation we rename \(S^i_{\text{QM}}\,\rightarrow\,S^i\).

In the sprit of the usual cosine formula \enquote{scalar product divided by norms} we normalize the spin flux by the norm of \(\vec{S}\) and divide by the area operator (we left out the \(x\) dependence to indicate the implicit sum over all \(x\in\mathcal{F}\)). Lastly, as \(S\cdot n_\mathcal{F}\) is only defined for a surface, we need to perform a limit which shrinks the surface to a point while keeping information about its orientation. Also, we symmetrize the naive expression such that the operator becomes self-adjoint.%
\footnote{The symmetrization procedure is in analogy to e.g. the cosine operator defined by S.\,Major in \cite{Major:1999mc} which we will also consider later on.}%
\begin{definition}
    We define the cosine operator at \(x_0\in\mathcal{F}\) to be 
    \begin{equation}
        \cos_{S\mathcal{F},x_0}=\frac{1}{2}\qty(C_{S\mathcal{F},x_0}+C_{S\mathcal{F},x_0}^\dagger) \,,\quad C_{S\mathcal{F},x_0}=\lim_{\epsilon\rightarrow0}\frac{1}{S[\mathcal{F}(\epsilon)]}S\cdot n_{\mathcal{F}(\epsilon)}\lim_{\epsilon\rightarrow0}\frac{1}{A[\mathcal{F}(\epsilon)]}
    \end{equation}
    where \(\mathcal{F}(\epsilon)\) is a disc centered around \(x_0\) with radius \(\epsilon\).
    The explicit formulas are
    \begin{equation}
        S[\mathcal{F}(\epsilon)]=\sum_{x\in\mathcal{F}(\epsilon)}\sqrt{\delta_{ij}S^i(x)S^j(x)}
    \end{equation}
    and
    \begin{equation}
        A[\mathcal{F}(\epsilon)]=\sum_{x\in\mathcal{F}(\epsilon)}\sqrt{\delta^{ij}\qty(J^{x,\mathcal{F}\uparrow}_i-J^{x,\mathcal{F}\downarrow}_i)\qty(J^{x,\mathcal{F}\uparrow}_j-J^{x,\mathcal{F}\downarrow}_j)} \,.
    \end{equation}
\end{definition}
Performing the limit we obtain an explicit form which allows for investigation on its spectrum:
\begin{align}
    C_{S\mathcal{F},x_0} &= \frac{1}{\sqrt{\delta_{ij}S^i(x_0)S^j(x_0)}}S^i(x_0)\qty(J^{x_0,\mathcal{F}\uparrow}_i-J^{x_0,\mathcal{F}\downarrow}_i) \times \\
    &\hphantom{{}={}\frac{1}{\sqrt{\delta_{ij}S^i(x_0)S^j(x_0)}}} \times\frac{1}{\sqrt{q^{ij}\qty(J^{x_0,S\uparrow}_i-J^{x_0,\mathcal{F}\downarrow}_i)\qty(J^{x_0,S\uparrow}_j-J^{x_0,\mathcal{F}\downarrow}_j)}} \\
    &= \frac{1}{\sqrt{\vec{S}^2(x_0)}}S^i(x_0)\qty(J^{x_0,\mathcal{F}\uparrow}_i-J^{x_0,\mathcal{F}\downarrow}_i)\frac{1}{A_{\mathcal{F},x_0}} .
\end{align}

%Similar to the classical case we we consider a disc \(\mathcal{F}(\epsilon)\) of radius \(\epsilon\) centered around \(x\in\Sigma\) and now shrink the surface to a point. We define
%\begin{equation}
%    \hat{C}(\alpha_{S\mathcal{F},x})\coloneqq\lim_{\epsilon\rightarrow0}\frac{\hat{S}[\mathcal{F}(\epsilon)]}{\sqrt{\hat{\vec{S}}^2}\hat{A}[\mathcal{F}(\epsilon)]}=\frac{1}{\sqrt{\hat{\vec{S}}^2(x)}}\hat{S}^i(x)\qty(\hat{J}^{x,S\uparrow}_i-\hat{J}^{x,S\downarrow}_i)\frac{1}{\sqrt{-\hat{\Delta}_{\mathcal{F},x}}}
%\end{equation}
%where
%\begin{equation}
%    -\hat{\Delta}_{\mathcal{F},x}=\sqrt{q^{ij}\qty(\hat{J}^{x,\mathcal{F}\uparrow}_i-\hat{J}^{x,\mathcal{F}\downarrow}_i)\qty(\hat{J}^{x,\mathcal{F}\uparrow}_j-\hat{J}^{x,\mathcal{F}\downarrow}_j)}
%\end{equation}
%is the usual area operator. The above is well defined if \(x\) is in \(\mathcal{F}\) and if the denominator does not vanish. As the basis in which the area operator is diagonal is not the same as the one in which the numerator is diagonal, we symmetrize it as
%\begin{equation}    
%    \widehat{\cos}(\alpha_{S\mathcal{F},x})=\frac{1}{2}\qty(\hat{C}(\alpha_{S\mathcal{F},x})+\hat{C}(\alpha_{S\mathcal{F},x})^\dagger) \,.
%\end{equation}
%
\subsection{Action of the cosine operator on spin network functions}
As the \(J_i^{v,p_j}\) only act nontrivially on the gravity tensor factor in the Hilbert space and \(S^i(v)\) only acts on the fermionic tensor factor, we can split the analysis of \(\cos_{S\mathcal{F}}\) into two parts.

First, we investigate the action of the spin operator on a general cylindrical function. For simplicity we assume that the spin network graph of some spin network \(\Psi\) only intersects the surface once at \(v\), furthermore we assume that \(p_j\) starts at \(v\), then the action of the spin operator is given by
\begin{equation}
    J_i^{v,p_j}\Psi(A)=\mathrm{i}\hbar\left.\frac{\dd}{\dd s}\right|_{s=0}\Psi(h_{e_1}(A),\ldots,h_{e_j}(A)\mathrm{e}^{s\tau_i},\ldots,h_{e_n}(A)) \,.
\end{equation}
We can rewrite the action on spin network functions \(T_{\gamma,\vec{j},\vec{\iota}}\) by abbreviating
\begin{equation}
    \tensor{\qty[\overset{(j_{e_i})}{\pi}(h_{e_i}(A))]}{^{m_i}_{n_i}}\equiv\tensor{\qty[\pi(h)]}{^{m_i}_{n_i}} \,,\quad h_i\equiv h_{e_i}(A) \,,
\end{equation}
as
\begin{align}
    J_i^{v,p_j}T_{\gamma,\vec{j},\vec{\iota}}(A) &=
    \ldots\cdot\qty(\mathrm{i}\hbar\left.\frac{\dd}{\dd s}\right|_{s=0}\tensor{\qty[\pi(h_j)\mathrm{e}^{s\tau_i}]}{^{m_j}_{n_j}})\cdot\ldots\cdot
    (\iota_v)^{n_1\dots n_j\dots n_n} \\
    &= 
    \ldots\cdot\tensor{\qty[\pi(h_j)]}{^{m_j}_{n_j}}\cdot\ldots\cdot
    \qty(\mathrm{i}\hbar\tensor{\qty[\pi(\tau_i)]}{^{n_j}_{n^\prime_j}}\,(\iota_v)^{n_1\dots n^\prime_j\dots n_n}) \,.
\end{align}
A similar trick can be performed for the action of the spin operator, namely we have 
\begin{equation}
    S^i(v)\qty(\ldots\cdot(\iota_{v^\prime})_{m_1\ldots m_na} \, (c^\dagger)^{a}(v^\prime)\ket{0})
    = \ldots\cdot \qty(\delta_{v,v^\prime}\,\mathrm{i}\hbar\,\tensor{{\tau_i}}{^b_a}\,(\iota_v)_{m_1\ldots m_nb})(c^\dagger)^{a}(v)\ket{0} \,.
\end{equation}
First, note that at the level of the intertwiners, \(J^{v,p}\) and 
\(S^i(v)\) act in exactly the same way on the indices of \(\iota_v\).%
\footnote{The fact that \(S_i\) acts as a \(J_i\) where the edge associated to the operator is spin \(\frac{1}{2}\) is somewhat expected as fermions couple to gravity via virtual spin \(\frac{1}{2}\) edges.} %
Hence we can naturally include the fermionic spin operator into the \(SU(2)\) graphical calculus as e.g. introduced in \cite{alesci2023graphicalcalculusspinnetworks}. Secondly, note that the objects in the bracket are no intertwiners purely in the \(m\) and \(a\) indices anymore as they are not invariant under gauge transformations%
\footnote{In fact, as the index \(i\) suggests, these objects transform in the \(j=1\) representation of \(SU(2)\).}%
. This is however not a problem as we will never apply the two operators alone but only their square \(\delta^{ik}J^{v,p}_iJ^{v,p}_k\) or \(\delta_{ik}S^i(v)S^k(v)\). This means \(\iota_v\) is now acted on by \(\pi^{(j)}(\tau_i)^2=-j(j+1)\mathbb{1}\) which clearly maps intertwiners to intertwiners. However, care needs to be taken when acting with the square of the sum of two operators belonging to different edges \(e_p\) and \(e_q\) (or to the virtual fermion edge and another edge \(e_q\)). In this case the action is diagonal only in a basis where \(j_p\) (or \(\frac{1}{2}\)) and \(j_q\) are coupled to some \(k\), with eigenvalue \(\hbar^2 k(k+1)\). On a technical level, this means we might have to recouple the intertwiner in order to change the coupling scheme using the formulas e.g. presented in \cite{alesci2023graphicalcalculusspinnetworks}. 

As we have just shown, it suffices to work only at the level of the intertwiners which is why these will be our central objects of study, both for calculating the spectrum of the cosine operator and for constructing the Bell test setup.
\subsection{Spectrum of the cosine operator}\label{sec:spec cos}
\begin{figure}
    \centering
    \begin{subfigure}[t]{0.47\textwidth}
         \centering
         \begin{tikzpicture}[scale=2]
            \filldraw[fill=gray!80,draw=gray!80] (-1.2,-0.01) rectangle (1.2,0.01);
            \node[circle,fill=black,draw,scale=0.25] (bullet) at (0,0) {};
            \foreach \a in {0,10,12,15,25,45,50,60,70,85,95,97,100,120,125,140,155,165,167,180,-20,-24,-35,-37,-50,-85,-62,-70,-80,-87,-92,-98,-105,-110,-112,-140,-145,-157,-165,-170,-173,-178}{
            \draw[gray!40!black,-] (bullet) -- (\a:1);
            }
            \node (S) at (-1.1,0.1) {\(\color{gray!80} S\)};
            \node[below left = 0.25 and 0.1] {\(v\)};
         \end{tikzpicture}
         \caption{General links at an arbitrary vertex \(v\).}
         \label{fig:udt_general}
     \end{subfigure}
     \hfill
     \begin{subfigure}[t]{0.47\textwidth}
         \centering
         \begin{tikzpicture}[scale=2]
            \filldraw[fill=gray!80,draw=gray!80] (-1.2,-0.01) rectangle (1.2,0.01);
            \node (S) at (-1.1,0.1) {\(\color{gray!80} S\)};
            \node[circle,fill=black,draw,scale=0.25] (bullet) at (0,0) {};
            \node[below left = -0.05] {\(v\)};
            \draw[thick,-]  (bullet) -- (0,1) node[above]{\(j^\uparrow\)};
            \draw[thick,-]  (bullet) -- (0,-1) node[below]{\(j^\downarrow\)};
            \draw[thick,-]  (bullet) -- (1,0) node[above right]{\(j^\rightarrow\)};
         \end{tikzpicture}
         \caption{Recoupled links according to their position relative to the surface \(S\).}
         \label{fig:udt_definition}
     \end{subfigure}
     \caption{Chosen recoupling prescription in order to calculate the action of the cosine operator at a vertex \(v\).}
\end{figure}
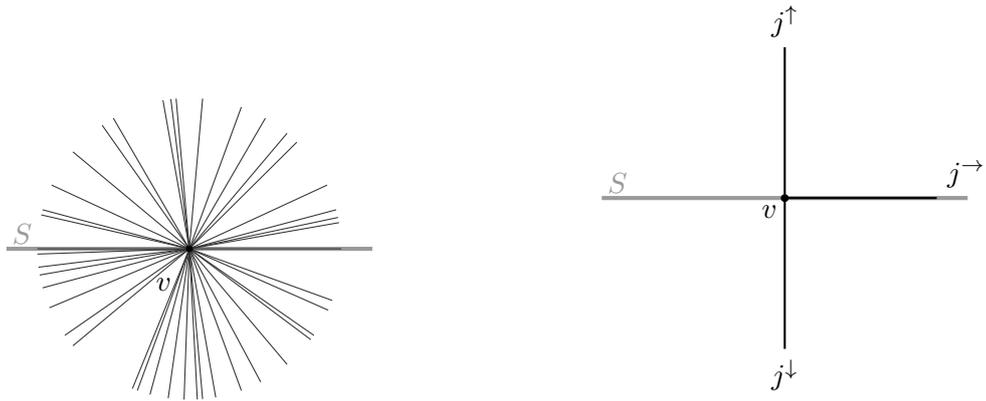

For calculating the spectrum we fix a vertex \(v\) and a surface \(S\) containing only this vertex for the rest of this subsection. We choose the coupling scheme
\begin{equation}
    \ket{l}\coloneqq\scoupling{\uparrow}{0.5}{\downarrow}{\rightarrow}[l] \,.
\end{equation}
We use the conventions of \cite{alesci2023graphicalcalculusspinnetworks}, however we work with normalized intertwiners. In order not having to deal with the degenerate case of a one dimensional Hilbert space we exclude the two cases where \(j^{\rightarrow}+\frac{1}{2}=\qty|j^{\uparrow}-j^{\downarrow}|\) and \(j^{\rightarrow}-\frac{1}{2}=j^{\uparrow}+j^{\downarrow}\). The advantage of this scheme is the fact that the area operator is diagonal in this basis. Upon rewriting the numerator of \(C_{S\mathcal{F},v}\) as sums of squares of angular momentum operators and using recoupling theory one finds that 
\begin{equation}\label{eq:pm_cos_pm}
    \mel**{\pm}{\cos_{S\mathcal{F},v}}{\pm}=\mp\frac{1}{\sqrt{3}}\frac{C}{\sqrt{B(\pm1)}}
\end{equation}
and
\begin{equation}\label{eq:pm_cos_mp}
    \mel**{\pm}{\cos_{S\mathcal{F},b}}{\mp} = -\frac{1}{\sqrt{3}}\sqrt{A-C^2}\qty(\frac{1}{2}\frac{1}{\sqrt{B(+1)}}+\frac{1}{2}\frac{1}{\sqrt{B(-1)}}) .
\end{equation}
We defined for later convenience 
\begin{align}
    A &\coloneqq \frac{1}{2}\qty(\frac{J^\uparrow}{J^\rightarrow})^2+\frac{1}{2}\qty(\frac{J^\downarrow}{J^\rightarrow})^2-\frac{1}{4} , \\
    B(r) &\coloneqq \frac{1}{2}\qty(\frac{J^\uparrow}{J^\rightarrow})^2+\frac{1}{2}\qty(\frac{J^\downarrow}{J^\rightarrow})^2-\frac{1}{4}-\frac{r}{4}\frac{1}{J^\rightarrow}-\frac{1}{4}\frac{1}{{J^\rightarrow}^2} , \\
    C &\coloneqq \frac{1}{2}\qty(\frac{J^\uparrow}{J^\rightarrow})^2-\frac{1}{2}\qty(\frac{J^\downarrow}{J^\rightarrow})^2
\end{align}
and \(J^{\uparrow,\downarrow,\rightarrow}\coloneqq j^{\uparrow,\downarrow,\rightarrow}+\frac{1}{2}\). As the cosine operator is symmetric \cref{eq:pm_cos_mp} has to be real, we conclude that \(A\geq C^2\). 

Calculating the spectrum is now only a matter of linear algebra.
\begin{prop}
    The eigenvalues of \(\cos_{S\mathcal{F},v}\) are 
    \begin{equation}
        \lambda_\pm=\pm\frac{1}{\sqrt{3}}\qty[\frac{1}{\sqrt
        {B(+1)}}\frac{\sqrt{A}\mp C}{2}+\frac{1}{\sqrt{B(-1)}}\frac{\sqrt{A}\pm C}{2}] \,.
    \end{equation}
\end{prop}

First, one might notice that the eigenvalues are a sum of two semi-positive terms as \(\sqrt{A}\geq \abs{C}\). Hence the kernel of the operator comprises of those states for which \(C=0\) and \(A=0\), i.e. \(2j^\uparrow=2j^\downarrow=j^\rightarrow\).

In the limit of large \(j^\uparrow,j^\downarrow\gg j^\rightarrow\) we have \(B(+1),B(-1)\approx A\) and thus \(\lambda_\pm=\pm\frac{1}{\sqrt{3}}\). This result matches exactly the two angles one would expect from a purely quantum mechanical calculation of the eigenvalues of the (normalized) spin operator projected onto an arbitrary direction. In fact, these are the maximally allowed angles which follows from \(B(+1)\leq B(-1)\) and the fact that in our case \(j^\rightarrow\leq j^\uparrow+j^\downarrow+\frac{1}{2}\).
\begin{prop}
The eigenvalues of \(\cos_{S\mathcal{F},v}\) are bound from above and below as
    \begin{equation}
        \frac{1}{\sqrt{3}}\sqrt{\frac{A}{B(-1)}} \leq \abs{\lambda_\pm} \leq \frac{1}{\sqrt{3}}\sqrt{\frac{A}{B(+1)}}\leq\frac{1}{\sqrt{3}} \,.
    \end{equation}
\end{prop}

For constructing the Bell test setup we will not consider the above general case but only \(j^{\uparrow,\downarrow}\neq0\), \(j^{\rightarrow}=0\). Now the intertwiner at \(x\) is unique and we must have \(j^{\downarrow}=j^{\uparrow}+\frac{p}{2}\). The cosine operator acts by multiplication with 
\begin{equation}\label{eq:sqrt3_cos}
        \frac{p}{\sqrt{3}}\frac{J^{\uparrow}+\frac{p}{4}}{\sqrt{(J^{\uparrow})^2+\frac{p}{2}J^{\uparrow}-\frac{5}{16}}} \quad\xrightarrow{\quad J^\uparrow\gg\frac{1}{2} \quad}\quad \frac{p}{\sqrt{3}}\qty(1+\frac{3}{16}\frac{1}{(J^{\uparrow})^2}) \,.
    \end{equation}
In the context of spin foams, the limit \(j\rightarrow\infty\), \(\hbar\rightarrow0\) with fixed \(\hbar j\) can be considered the semiclassical limit of the spinfoam amplitude, see for example \cite{Barrett:2009gg}.   We will also work in the regime \(j^{\uparrow}\gg\frac{1}{2}\) which leads us to not work with the full cosine operator but only with its sign. We define the corresponding operator via its action on the states.
\begin{definition}
    Consider a vertex \(v\) with intertwiner \(\iota_v\) of a spin network and a surface \(S\) intersecting the vertex such that no edges are tangential to \(S\). Recouple all the spins associated to edges lying above, below \(S\) to \(j^\uparrow\), \(j^\downarrow\) respectively. We must have that \(j^\downarrow=j^\uparrow+\frac{p}{2}\). Write \(\ket{k_1,\ldots,j^\uparrow,j^\uparrow+\frac{p}{2},\ldots,k_n}\) for the intertwiner where the \(k_i\) are the internal spins specifying the chosen basis. The sign operator is defined to be the sign of the cosine operator
    \begin{equation}
        \operatorname{sgn}_{S\mathcal{F},v} \coloneqq \operatorname{sgn}(\cos_{S\mathcal{F},v}) \,,
    \end{equation}
    it acts on the intertwiner as
    \begin{equation}\label{eq:sgn operator definition}
        \operatorname{sgn}_{S\mathcal{F},v} \ket{k_1,\ldots,j^\uparrow,j^\uparrow+\frac{p}{2},\ldots,k_n} = p \ket{k_1,\ldots,j^\uparrow,j^\uparrow+\frac{p}{2},\ldots,k_n} \,.
    \end{equation}
    If the recoupling spins are not relevant for a calculation we drop them and write \(\ket{p}\equiv\ket{k_1,\ldots,j^\uparrow,j^\uparrow+\frac{p}{2},\ldots,k_n}\).
\end{definition}

Note that the sign operator is a self-adjoint, trace-free operator on \(\mathbb{C}^2\simeq\operatorname{span}(\ket{+},\ket{-})\). It may thus be written as \(\vec{n}\cdot\vec{\sigma}\). In the basis \(\ket{\pm}\) it acts just as \(\sigma_z\) (or \(-\sigma_z\)).

In the large \(j\) limit, the sign operator is nothing but the quantum mechanical spin projection operator (up to normalization), hence its eigenvalue tells us about the orientation of the fermion spin with respect to the surface normal.
%As in the general case, in the limit of large \(j^{\uparrow}\) we arrive at the usual quantum mechanics cosine operator:
%\begin{equation}\label{eq:sqrt3_cos}
%    \sqrt{3}\widehat{\cos}(\alpha_{\mathcal{SF},x}) \ket{p} \approx p \qty(1+\frac{3}{16}\frac{1}{(J^{\uparrow})^2}) \ket{p} \,.
%\end{equation}
%But this is, up to quadratic order, just the signum of the eigenvalue. 

%In the following we go over to the semi-classical limit%
%\footnote{\textcolor{red}{Justification necessary?}} %
%\(j^{\uparrow}\gg\frac{1}{2}\) which leads us to not work with the full cosine operator but only with its sign. We define the corresponding operator via its action on the states as
%\begin{equation}\label{eq:sgn operator definition}
%    \widehat{\operatorname{sgn}}_{S\mathcal{F},x} \ket{p} \coloneqq p \ket{p} \,.
%\end{equation}
%In case of \(j^{\downarrow}=0\) we have \(j^{\rightarrow}=j^{\uparrow}+\frac{p}{2}\) and again the operator acts by multiplying with  
%\begin{equation}
%    \cos(\alpha_{S\mathcal{F},x_0})=\frac{p}{\sqrt{3}}\frac{J^{\uparrow}-\frac{p}{2}}{\sqrt{(J^{\uparrow})^2-\frac{1}{4}}} \xrightarrow{j^{\uparrow}\gg\frac{1}{2}} \frac{p}{\sqrt{3}} \,.
%\end{equation}
%Lastly, in case of \(j^{\rightarrow}=0\) and  we simply obtain
%\begin{equation}
%    \cos(\alpha_{S\mathcal{F},x_0})=-1 \,,
%\end{equation}
%intuitively the spin has to point down as the closure condition \(\hat{\vec{J}}^{\uparrow}+\hat{\vec{S}}=0\) demands it. 
%
\section{Bell states}
Having the sgn operator at our disposal we are finally able to mimic the famous Bell test for entanglement between two fermions in a gauge invariant setting of loop quantum gravity. To this end, we first recall the original ideas of Bell and Clauser, Horne, Shimony, Holt (CHSH). 
\label{sec:bell}
\subsection{Classcial Bell argument and CHSH inequality}\label{sec:class Bell inequ}
Consider two observers Alice (A) and Bob (B), possibly spacelike separated, which each have one detector at their disposal. A pair of entangled qubits is created and each of the particles is sent to one of the observer. Each A and B can perform a measurement \textit{independent of the measurement outcome of the partner} based on one of two detector settings with the only two outcomes of one fixed setting being \(+1\) or \(-1\). In our case the measurement consists of measuring the projection of the spin of the particle onto one of two directions (labelled by \(i\)) \(\vec{n}^\text{A}_i\), \(\vec{n}^\text{B}_i\) of their choice. We denote the Hilbert space which we use to describe the particle and the measurement at A, B  by \(\mathcal{H}_{\text{A},\text{B}}=\mathbb{C}^2\). Here we simplify the situation and consider \(\vec{n}^\text{A}_0=\vec{n}^\text{B}_0=\vec{n}_0\). We write \(A_{0,1}\) and \(B_{0,1}\) for the operator describing the measurement of A's or B's detector into the direction 0 and 1 respectively. 

Consider now the expectation value of the CHSH correlation operator
\begin{equation}
    B=A_0\otimes B_0+A_0\otimes B_1+A_1\otimes B_0-A_1\otimes B_1
\end{equation}
on \(\mathcal{H}=\mathcal{H}_\text{A}\otimes\mathcal{H}_\text{B}\). The CHSH inequality first derived in \cite{Clauser:1969ny} states that \(|\langle B\rangle|\leq2\) if one assumes only classical correlations. By that we mean correlations where the measurement outcome is determined by the detector setting and a (so-called hidden) variable \textit{a priori} to the actual measurement. Famously, one can construct quantum states for which \(|\langle B\rangle|>2\). As those pairs of qubits violate the CHSH inequality they have to be quantum mechanically entangled. For a concrete example of CHSH violation we work in the Bell basis given by
\begin{align}
    \ket{\Phi^{\pm}} &= \frac{1}{\sqrt{2}}\qty(\ket{+}_\text{A}\otimes\ket{+}_\text{B}\pm\ket{-}_\text{A}\otimes\ket{-}_\text{B}) \\
    \ket{\Psi^{\pm}} &= \frac{1}{\sqrt{2}}\qty(\ket{+}_\text{A}\otimes\ket{-}_\text{B}\pm\ket{-}_\text{A}\otimes\ket{+}_\text{B}) \,.
\end{align}
Consider the case of \(A_0=B_0=-\sigma^z\) and \(A_1=A^z\sigma_z+A^x\sigma_x\), \(B_1=B^z\sigma_z+B^x\sigma_x\). Then there exist eigenvectors of the CHSH operator \(B\), namely
\begin{equation}
        \frac{1}{\sqrt{1+{\chi_{\mp}}^2}}\qty(\ket{\Phi^{-}}-\chi_{\mp}\ket{\Psi^{+}})
    \end{equation}
    where
    \begin{equation}
        \chi_{\pm}=\frac{A^z B^z+A^z+B^z-2\sqrt{1+\abs{A^x}\abs{B^x}}-(1+\abs{A^x} 
        \abs{B^x})}{B^x(1+A^z)+A^x(1+B^z)} ,
    \end{equation}
whose eigenvalues are \(\lambda=\pm2\sqrt{1+\abs{A^x}\abs{B^x}}\), \(\abs{\lambda}\geq2\). Hence a measurement in this state would yield an outcome which clearly cannot be explained by hidden variables but only by entanglement of the qubit pair. For a more extensive analysis of the eigenvectors and eigenvalues of the CHSH operator we refer to \cite{CHEFLES19974}.

Following the idea of the Bell inequality violation implying entanglement, the goal of the following sections is to explicitly construct the Hilbert spaces \(\mathcal{H}_{\text{A},\text{B}}\) and the operators \(A_{0,1}\), \(B_{0,1}\) in the framework of loop quantum gravity.
\subsection{Three-valent LQG setup}\label{sec:3valent setup}
The main insight of this section is the simple correspondence between the \(\ket{\pm}\) states of the Bell test and of the sign operator introduced in \cref{eq:sgn operator definition}. We realize the \(\mathbb{C}^2\) intertwiner states by first coupling at least three external spins \(j^{\text{ext}}_1,\ldots,j^{\text{ext}}_n\) with the fermion spin to intertwiners in \(\operatorname{Inv}(\mathcal{H}_{j^{\text{ext}}_1}\otimes\ldots\otimes\mathcal{H}_{j^{\text{ext}}_n}\otimes\mathcal{H}_{\frac{1}{2}})\). We label the basis elements of this vector space by the internal coupling spins \(k_i\). As this space is very high dimensional we fix as many internal spins as necessary in order to arrive at \(\operatorname{Inv}(\mathcal{H}_{k_1}\otimes \mathcal{H}_{k_2}\otimes \mathcal{H}_{k_3}\otimes\mathcal{H}_\frac{1}{2})\simeq\mathbb{C}^2\). We exclude the case of \(k_3\pm\frac{1}{2}=\abs{k_1\pm k_2}\) as the space would otherwise be only one dimensional. A natural coupling scheme which tells us which \(k_i\) to choose is determined by the orientation of the edges intersecting a vertex with respect to a given surface. This is motivated by the fact that the sign operator is diagonal in a basis where all the spins associated to the edges coming from above, below are recoupled to \(j^\uparrow\), \(j^\downarrow\) respectively (we assume \(j^\rightarrow=0\) in the following, otherwise one would \enquote{tweak} the surface such that this requirement holds). Introducing not just one but two surfaces \(S_0\) and \(S_1\) at a fixed vertex \(v\) partitions the set of edges meeting at \(v\) into four regions \(I\), \(II\), \(III\), \(IV\). See \cref{fig:2_surfaces_definition} for a graphical representation. Naturally, we choose a coupling scheme where all the spins in one wedge are recoupled to \(j_1\), \(j_2\), \(j_3\) or \(j_4\) accordingly. Even after fixing all the internal spins in every wedge, \(\operatorname{Inv}(\mathcal{H}_{j_1}\otimes \mathcal{H}_{j_2}\otimes \mathcal{H}_{j_3}\otimes\mathcal{H}_{j_4}\otimes\mathcal{H}_\frac{1}{2})\) is in general still too high dimensional. The two ways out are to set one of the spins to zero, as done in this first subsection, or to fix one of the remaining recoupling spins, as done in the next subsection.
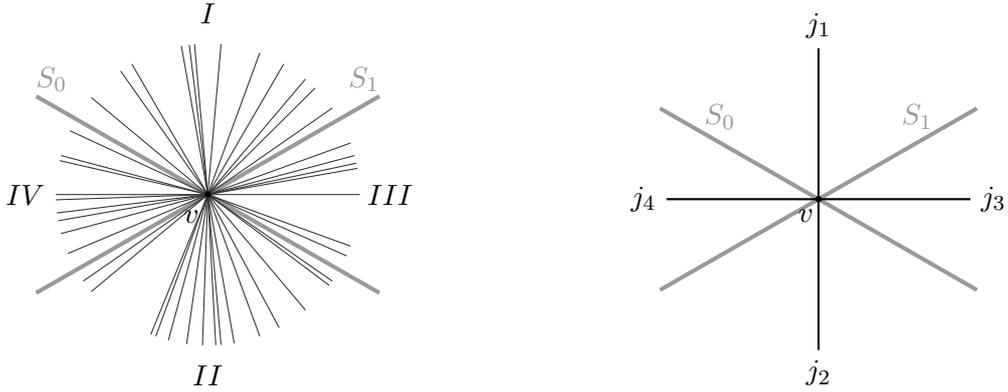
\begin{figure}
    \centering
    \begin{subfigure}[b]{0.47\textwidth}
        \centering
        \begin{tikzpicture}[scale=2]
          \coordinate (bullet) at (0,0);
          \filldraw[fill=gray!80,draw=gray!80,rotate around={30:(bullet)}] (-1.3,-0.01) rectangle (1.3,0.01);
          \filldraw[fill=gray!80,draw=gray!80,rotate around={-30:(bullet)}] (-1.3,-0.01) rectangle (1.3,0.01);
          \node[above = 0] (S0) at (149:1.2) {\(\color{gray!80} S_0\)};
          \node[above = 0] (S1) at (31:1.2) {\(\color{gray!80} S_1\)};
          \foreach \a in {0,10,12,15,20,35,45,50,60,70,85,95,97,100,120,125,140,155,165,167,180,-20,-24,-35,-37,-50,-85,-62,-70,-80,-87,-92,-98,-105,-110,-112,-140,-145,-157,-165,-170,-173,-178}{
            \draw[gray!40!black,-] (bullet) -- (\a:1);
          }
          \draw[fill=black,draw=black] (bullet) circle (0.5pt);
          \node[below left = 2pt and -0.5pt of bullet] {\(v\)};
          \node (I) at (90:1.2) {\contour{white}{\(I\)}}; 
          \node (II) at (-90:1.2) {\contour{white}{\(II\)}}; 
          \node (III) at (0:1.2) {\contour{white}{\(III\)}}; 
          \node (IV) at (180:1.2) {\contour{white}{\(IV\)}}; 
        \end{tikzpicture}
        \caption{General links at an arbirary vertex \(v\) with two surfaces \(S_{0,1}\).}
        \label{fig:2_surfaces_general}
     \end{subfigure}
     \hfill
     \begin{subfigure}[b]{0.47\textwidth}
        \centering
        \begin{tikzpicture}[scale=2]
          \coordinate (bullet) at (0,0);
          \filldraw[fill=gray!80,draw=gray!80,rotate around={30:(bullet)}] (-1.2,-0.01) rectangle (1.2,0.01);
          \filldraw[fill=gray!80,draw=gray!80,rotate around={-30:(bullet)}] (-1.2,-0.01) rectangle (1.2,0.01);
          \node[above = 0.05] (S) at (150:0.75) {\(\color{gray!80} S_0\)};
          \node[above = 0.05] (S) at (30:0.75) {\(\color{gray!80} S_1\)};
          \draw[fill=black,draw=black] (bullet) circle (0.5pt);
          \node[below left = 0.1pt and -2pt of bullet] {\(v\)};
          \draw[thick,-]  (bullet) -- (0,1) node[above]{\(j_1\)};
          \draw[thick,-]  (bullet) -- (0,-1) node[below]{\(j_2\)};
          \draw[thick,-]  (bullet) -- (1,0) node[right]{\(j_3\)};
          \draw[thick,-]  (bullet) -- (-1,0) node[left]{\(j_4\)};
        \end{tikzpicture}
        \caption{Recoupled links according to their position relative to the two surfaces \(S_{0,1}\).}
        \label{fig:2_surfaces_definition}
     \end{subfigure}
     \caption{Recoupling prescription Alice and Bob choose in order to calculate the action of the cosine operator associated to \(S_0\) and \(S_1\) respectively. These and similar figures are inspired by \cite{Major:1999mc}.}
\end{figure}
\subsubsection{Exact solution}
In this subsection we assume \(j_4=0\). We assign to each Alice and Bob such a three-valent vertex \(v_{\text{A},\text{B}}\) of a spin network and two surfaces \(S_{0,1}\) whose normals \(n_{0,1}\) played the role of the detector settings in the classical Bell setup. In the following we label the recoupled spins in the four regions at Alice's or Bob's vertex as \(j^{\text{A},\text{B}}_i\) (and \(J^{\text{A},\text{B}}_i\) respectively). If we omit the labels A and B in an expression, we mean that this equation holds for both both \(j^\text{A}\) and \(j^\text{B}\). The two observers now measure the sign of the cosine operator associated to each surface via \(\operatorname{sgn}^{\text{A},\text{B}}_{SS_i,v_{\text{A},\text{B}}}\) as defined in \cref{sec:spec cos}. In the language of \cref{sec:class Bell inequ}:
\begin{equation}
    A_{0,1}=\operatorname{sgn}^{\text{A}}_{SS_{0,1},v_\text{A}} \,,\qquad B_{0,1}=\operatorname{sgn}^{\text{B}}_{SS_{0,1},v_\text{B}} \,.
\end{equation}
For the following calculations we fix two bases for both Alice and Bob. In order to lighten the notation we drop the labels \(A\), \(B\) and \(v\) of the sign operator represented in these bases.

Since the result will not depend on the initial coupling scheme we let both Alice and Bob start in an eigenbasis for \(\operatorname{sgn}_{SS_0}\). Namely they recouple \(j_1\) and \(j_3\) to \(j_{13}=j_2+\frac{r}{2}\) and set 
\begin{equation}
    \ket{\pm}_0\equiv\tcoupling{1}{3}{0.5}{2}[j_2\pm\frac{1}{2}] .
\end{equation}
As done for the spins, we will drop the label \(A\) or \(B\) if the result holds for both \(\ket{r}^\text{A}\) and \(\ket{r}^\text{B}\). Clearly, \(\operatorname{sgn}_{SS_1}\) is not diagonal in this basis, so A and B will have to recouple to intertwiners \(\ket{r}_1\) where \(j_2\) and \(j_3\) are coupled to \(j_{23}=j_1+\frac{r}{2}\). As the recoupling is a unitary change of bases and as we choose the Condon Shortley phase convention where all Clebsch-Gordan coefficients are real, we may write
\begin{equation}
    \ket{r}_1=\sum_s O_{rs}\ket{s}_0 \quad,\quad \ket{r}_0=\sum_s(O^T)_{rs}\ket{s}_1
\end{equation}
with \(O^T=O\in O(2)\)%
\footnote{The symmetry of \(O\) follows from the recoupling theory of the states and the property of the \(6j\) symbol.}. %
We know that \(\operatorname{sgn}^{(1)}_{SS_1}=\sigma_z=-\operatorname{sgn}^{(0)}_{SS_0}\)%
\footnote{Note that for \(\ket{\pm}_0\) we have \(j_2=j^\downarrow\) and \(j_2\pm\frac{1}{2}=j^\uparrow\) which results in the additional minus sign.} %
where the superscript indicates the basis. Changing into the \(0\) basis via \(O\) results in
\begin{equation}\label{eq:sgn_1 cos sin}
    \operatorname{sgn}^{(0)}_{SS_1}=-O\operatorname{sgn}^{(0)}_{SS_0}O .
\end{equation}
This is in fact basis independent as changing into another basis labeled \(2\) via the transformation \(\tilde{O}\) from 0 to 2 we get, using that \([\tilde{O},O]=0\):
\begin{equation}
    \operatorname{sgn}^{(2)}_{SS_1}=\tilde{O}\operatorname{sgn}^{(0)}_{SS_1}\tilde{O}=\tilde{O}O\operatorname{sgn}^{(1)}_{SS_1}O\tilde{O}=-O\tilde{O}\operatorname{sgn}^{(0)}_{SS_0}\tilde{O}O=-O\operatorname{sgn}^{(2)}_{SS_0}O \,,
\end{equation}
in other words we have in any basis
\begin{equation}\label{eq:sgn_1 sgn_0}
    \operatorname{sgn}_{SS_1}=-O \operatorname{sgn}_{SS_0} O \,.
\end{equation}
As every symmetric \(O\in O(2)\) can be written as a reflection followed by a rotation \(\exp(\mathrm{i}\frac{\alpha}{2}\sigma_y)\sigma_z\) we may write, introducing \(\operatorname{sgn}_{SS_{0,1}}=\vec{n}_{0,1}\cdot\vec{\sigma}\),
\begin{equation}
    \vec{n}_1\cdot\vec{\sigma} = \exp\qty(-\mathrm{i}\frac{\pi+\alpha}{2}\sigma_y) \, \vec{n}_0\cdot\vec{\sigma} \, \exp\qty(\mathrm{i}\frac{\pi+\alpha}{2}\sigma_y) = (R_y(\alpha+\pi)\vec{n}_0)\cdot\vec{\sigma}
\end{equation}
where \(R_y(\alpha+\pi)\) describes a rotation around the \(y\) axis with angle \(\pi+\alpha\). From a quantum mechanics this is just the projection of the fermion spin (modulo \(\frac{\hbar}{2}\)) onto the two directions \(\vec{n}_{0,1}\). In other words, the fermion \enquote{measures} the angle between the two surface normals to be \(\alpha+\pi\). 

As \cref{eq:sgn_1 sgn_0} holds in any basis we evaluate it in the 0 basis and get
\begin{equation}
    \operatorname{sgn}^{(0)}_{SS_1}=\cos(\alpha)\sigma_z+\sin(\alpha)\sigma_x .
\end{equation}
But this is exactly the same expression as in \cref{sec:class Bell inequ} for \(A_1\), \(B_1\) which means we can write down eigenvectors of the correlation operator \(B\) with eigenvalues
\begin{equation}
    \abs{\lambda}=2\sqrt{1+\abs{\sin(\alpha_A)}\abs{\sin(\alpha_B)}}\geq2 .
\end{equation}
Clearly, those generalized spin network states exhibit correlations which are larger than hidden variables would allow them to be, hence Alice and Bob have shown that the two fermions have to be entangled. Curiously, the magnitude of \(\lambda\) only depends on the angles the fermion measures between the two surface normals at both A and B.

An explicit calculation involving \(6j\) symbols yields
\begin{equation}
    \cos(\alpha)=\sum_{r\in\{\pm\}} r\,\qty(O_{+r})^2=\frac{(J_1)^2+(J_2)^2-(J_3)^3}{2J_1J_2} \,.
\end{equation}
This can be written as
\begin{equation}
    \cos(\alpha)=\frac{\vec{J}_1\cdot\vec{J}_2}{\norm\big{\vec{J}_1}\norm\big{\vec{J}_2}} 
\end{equation}
with \(\alpha\) being the angle opposite to \(\vec{J}_3\) in a triangle formed by the \(\vec
J_i\) as shown in \cref{fig: triangle}.

Note that technically we have constructed specific  states which carry intertwiner entanglement in the sense of \cite{Livine:2017fgq,Chen:2022rty,Bianchi_2023}. 
This is necessary as the cosine operator involves geometric operators acting on the gravitational part of the generalized spin network state. Indeed, consider a quite generic setup where an operator is of the form \(a_{ij}J^i\otimes S^j\) with some \(J^i\) acting non trivially on the holonomies or intertwiners. Also here, we can, by duality, consider operators that do not contain holonomies as acting on the intertwiners which means the answer to questions about spatial geometry or the fermions will be given in terms of properties of the intertwiners. 
%This is somewhat natural as e.g. the only freedom in the construction \(\iota_{\ldots A_I \ldots}\,{c^\dagger}^{A_I}\) lies in \(\iota\).

Intertwiner entanglement has also been important in the construction of other states with special properties, see for example \cite{Freidel_2010,Bianchi:2018fmq,Baytas:2018wjd}.

\subsubsection{Continuum limit}
Often, as the classical limit in the spin network formalism one considers the coupling of many small \(j\)'s at the vertices. In this light it is natural to ask how the angle mentioned above changes when a small spin \(\Delta j\) edge changes its characteristic w.r.t. the surface \(S_1\). 
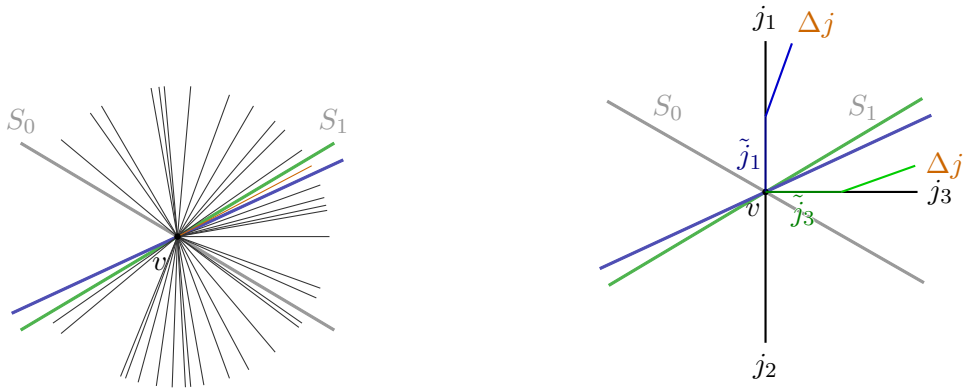
\begin{figure}
    \centering
    \begin{subfigure}[t]{0.47\textwidth}
        \centering
        \begin{tikzpicture}[scale=2]
          \coordinate (bullet) at (0,0);
          \filldraw[fill=gray!60!green,gray!60!green,rotate around={31:(bullet)}] (-1.2,-0.007) rectangle (1.2,0.007);
          \filldraw[fill=gray!60!blue,draw=gray!60!blue,rotate around={25:(bullet)}] (-1.2,-0.007) rectangle (1.2,0.007);
          \filldraw[fill=gray!80,draw=gray!80,rotate around={-31:(bullet)}] (-1.2,-0.007) rectangle (1.2,0.007);
          \node[above = 0] (S0) at (149:1.2) {\(\color{gray!80} S_0\)};
          \node[above = 0] (S1) at (31:1.2) {\(\color{gray!80} S_1\)};
          %\node[above = 0] (S1) at (26:1) {\(\color{gray!80} S_1\)};
          \foreach \a in {0,10,12,15,20,35,45,50,60,70,85,95,97,100,120,125,140,-20,-24,-35,-37,-50,-85,-62,-70,-80,-87,-92,-98,-105,-110,-112,-140,-145}{
            \draw[gray!40!black,-] (bullet) -- (\a:1);
          }
          \draw[orange!80!black,-] (bullet) -- (28:1);
          \draw[fill=black,draw=black] (bullet) circle (0.5pt);
          \node[below left = 4pt and -0.05pt of bullet] {\(v\)};
        \end{tikzpicture}
        \caption{The orange spin \(\Delta j\) edge changes its position w.r.t. the two possible orientations of the surface \(S_1\) at a vertex \(v\).}
        \label{fig:1plus2_surfaces_3j_general}
     \end{subfigure}
     \hfill
     \begin{subfigure}[t]{0.47\textwidth}
        \centering
        \begin{tikzpicture}[scale=2]
          \coordinate (bullet) at (0,0);
          \filldraw[fill=gray!60!green,draw=gray!60!green,rotate around={31:(bullet)}] (-1.2,-0.007) rectangle (1.2,0.007);
          \filldraw[fill=gray!60!blue,draw=gray!60!blue,rotate around={25:(bullet)}] (-1.2,-0.007) rectangle (1.2,0.007);
          \filldraw[fill=gray!80,draw=gray!80,rotate around={-30:(bullet)}] (-1.2,-0.007) rectangle (1.2,0.007);
          \node[above = 0.05] (S) at (149:0.75) {\(\color{gray!80} S_0\)};
          \node[above = 0.05] (S) at (31:0.75) {\(\color{gray!80} S_1\)};
          \draw[fill=black,draw=black] (bullet) circle (0.5pt);
          \node[below left = 0.5pt and -2pt of bullet] {\(v\)};
          \draw[thick,-]  (0,0.5) -- (0,1) node[above]{\(j_1\)};
          \draw[thick,blue!50!black,-]  (bullet) -- node[midway,left=-0.1]{\(\color{blue!50!black} \tilde{j}_1\)} (0,0.5);
          \draw[blue!80!black,thick,-] (0,0.5) -- (80:1) node[above right=-0.1] {\(\color{orange!80!black} \Delta j\)};
          \draw[thick,-]  (bullet) -- (0,-1) node[below]{\(j_2\)};
          \draw[thick,green!50!black,-]  (bullet) -- node[midway,below=-0.1]{\contour{white}{\(\textcolor{green!50!black}{\tilde{j}_3}\)}} (0.5,0);
          \draw[thick,-]  (0.5,0) -- (1,0) node[right]{\(j_3\)};
          \draw[green!80!black,thick,-] (0.5,0) -- (10:1) node[right] {\(\color{orange!80!black} \Delta j\)};
        \end{tikzpicture}
        \caption{Coupling of \(\Delta j\) to the above spin \(j_1\) or below spin \(j_3\) based on the orientation of \(S_1\).}
        \label{fig:1plus2_surfaces_3j_definition}
     \end{subfigure}
     \caption{Chosen recoupling prescription in case a small spin \(\Delta j\) edge changes its orientation w.r.t. \(S_1\).}
\end{figure}
Therefore we calculate the difference in angle between a setting where the \(\Delta j\) edge happens to lie below the surface \(S_1\) yielding \(a^{\downarrow}=\cos(\alpha^{\downarrow})\) and a setting where the \(\Delta j\) edge happens to lie above the surface yielding \(a^{\uparrow}=\cos(\alpha^{\uparrow})\). In the first case, \(j_3\) and \(\Delta j\) couple to \(\tilde{j}_3=j_3+\Delta j^{\downarrow}\) (we write \(\Delta J^{\downarrow}=\Delta j^{\downarrow}+\frac{1}{2}\)), in the second case \(j_1\) and \(\Delta j\) couple to \(\tilde{j}_1=j_1+\Delta j^{\uparrow}\) (we write \(\Delta J^{\uparrow}=\Delta j^{\uparrow}+\frac{1}{2}\)). For a pictorial presentation see \cref{fig:1plus2_surfaces_3j_general} and \cref{fig:1plus2_surfaces_3j_definition}. %
%a small spin \(\Delta j^{\downarrow}=\abs{\tilde{j}_3-j_3}\) (or \(\Delta J^{\downarrow}=\Delta j^{\downarrow}+\frac{1}{2}\)) happens to lie below the surface \(S_1\) yielding \(a^{\downarrow}=\cos(\theta^\downarrow)\) and a setting where \(\Delta j^{\uparrow}=\abs{\tilde{j}_1-j_1}\) (or \(\Delta J^{\uparrow}=\Delta j^{\uparrow}+\frac{1}{2}\)) happens to lie above \(S_1\) yielding \(a^{\uparrow}=\cos(\theta^{\uparrow})\). 
For the sake of clarity we only write down the difference in the two angles in the case of \(\Delta J^{\uparrow}=\Delta J^{\downarrow}\) \footnote{Here we have assumed that \(\operatorname{sgn}\qty(\sin\qty(\alpha^\uparrow))=\operatorname{sgn}\qty(\sin\qty(\alpha^\downarrow))\), i.e. the flip of the vertex does not appear near \(0,\pi\). A more general expression for the difference in angles we refer to \cite{Zeiss:2025}}:
\begin{equation}
    \alpha^{\downarrow}-\alpha^{\uparrow} \approx \frac{(J_1+J_3)^2-(J_2)^2}{\sqrt{\qty((J_1+J_2)^2-(J_3)^2)\qty((J_3)^2-(J_1-J_2)^2)}}\,\frac{\Delta J^{\uparrow}}{J_1} \,.
\end{equation} 
Clearly, in the regime of large \(J_1\) or small \(\Delta J\) this difference vanishes which could be interpreted as a continuum limit.
\subsubsection{Interpretation of the angle \(\mathbf{\alpha}\)}
\begin{figure}
    \centering
    \begin{subfigure}[t]{0.47\textwidth}
        \centering
        \begin{tikzpicture}[
        scale=1.5,
        >=latex,
        ]
            \coordinate (A) at (0,0);
            \coordinate (B) at (3,0);
            \coordinate (C) at (2,-1);
            \coordinate (C1) at (3,-1.5);
            \coordinate (D) at (1,2);
            \draw[thick,->] (A) -- node[below left=-0.1]{\(\vec{J}_1\)} (C);
            \draw[thick,->] (C) -- node[below right=-0.1]{\(\vec{J}_{24}\)} (B);
            \draw[thick,->] (B) -- node[pos=0.7,above=0.0]{\(\vec{J}_3\)}  (A);  
            \draw[thick,->] (A) -- node[above left=-0.1]{\(\vec{J}_{14}\)} (D);
            \draw[thick,->] (C) -- node[above right=-0.1]{\(\vec{J}_4\)} (D);
            \draw[thick,->] (D) -- node[above right=-0.1]{\(\vec{J}_2\)} (B);
        \end{tikzpicture}
        \caption{Case \(j_{1,2,3,4}\neq0\): \(\alpha\) is the angle between the two opposing sides \(J_{14}\) and \(J_{24}\) in the tetrahedron.}
        \label{fig: tetrahedron}
    \end{subfigure}
    \hfill
    \begin{subfigure}[t]{0.47\textwidth}
        \centering
        \begin{tikzpicture}[
            scale=2,
            >=latex
            ]
            \coordinate (A) at (0,0);
            \coordinate (B) at (3,0);
            \coordinate (C) at (2,-1);
            \draw[thick,->] (A) -- node[below left]{\(\vec{J}_1\)} (C);
            \draw[thick,->] (C) -- node[below right]{\(\vec{J}_2\)} (B);
            \draw[thick,->] (B) -- node[above=0.1]{\(\vec{J}_3\)}  (A);  
        \end{tikzpicture}
        \caption{Case \(j_4=0\): \(\alpha\) is the external angle between the two adjacent sides \(J_1\) and \(J_2\) in the triangle.}
        \label{fig: triangle}
    \end{subfigure}
    \caption{Geometric interpretation of the angle \(\alpha\) which describes the maximally attainable entanglement.}
\end{figure}
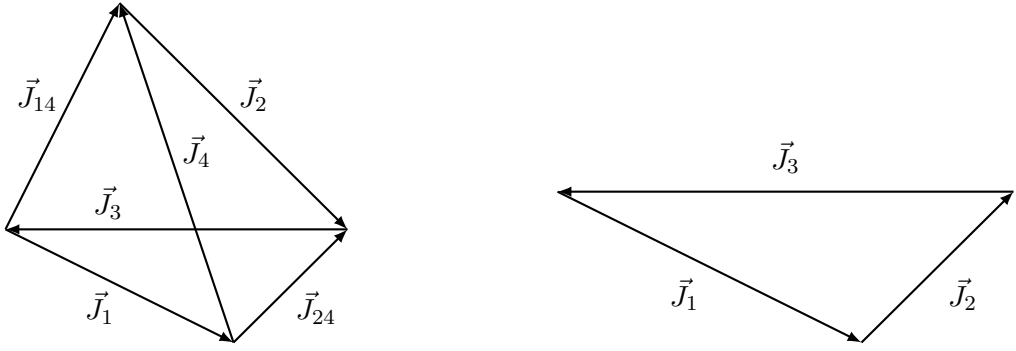
As we have established in \cref{sec:3valent setup} the fermion measures an angle \(\theta_{01}=\arccos(\vec{n}_0\cdot\vec{n}_1)=\arccos(-\cos(\alpha))=\pi+\alpha\) between the surface normals. It is natural to ask whether \(\theta_{01}\) agrees with the angle an outside observer would measure. Such an operator was defined by S.\,Major in \cite{Major:1999mc} and \cite{Major:2011ry}, based on the quantization of the scalar product of the two surface normals. We adjust the notation of \cite{Major:1999mc} and define
\begin{equation}
    \cos_{S_0S_1,v}=\frac{1}{2}\qty[C_{S_0S_1,v}+C_{S_0S_1,v}^\dagger]
\end{equation}
with
\begin{equation}
    C_{S_0S_1,v}=\frac{1}{\sqrt{A_{S_1,v}}}T_{S_1S_0,v}\frac{1}{\sqrt{A_{S_0,v}}} ,\quad T_{S_1S_0,v}=\delta^{ij}\qty(J_i^{v,S_1\uparrow}-J_i^{v,S_1\downarrow})\qty(J_j^{v,S_0\uparrow}-J_j^{v,S_0\downarrow}) .
\end{equation}
Here \(T_{S_1S_0,v}\) is the scalar product operator associated to the two surface normals. In the limit of \(j_i\gg\frac{1}{2}\) for all three vertices one finds after employing some recoupling theory
\begin{align}
    \mel**{\pm}{\cos_{S_0S_1,v}}{\pm} &\approx \frac{(J_1)^2+(J_2)^2-(J_3)^2}{2J_1J_2} \\
    \mel**{\pm}{\cos_{S_0S_1,v}}{\mp} &\approx 0
\end{align}
where we indicate by \(\approx\) that we neglect terms that fall off as \(\frac{1}{J_i}\) or \(\frac{J_i}{(J_j)^2}\) and higher orders thereof. Hence for large spins the operator acts diagonally on the \(\ket{r}\) with eigenvalue \(\cos(\alpha)\). We did not specify a certain basis as this holds in any basis since the diagonal operator is invariant under a \(O(2)\) trafo. A short calculation shows that this result is exact if we neglect the virtual fermion edge all along and only couple \(j_1\), \(j_2\) and \(j_3\) at \(v\). 
%\textcolor{red}{MORE INTERPETATION?}

In summary we find that the angle between the two surface normals measured with respect to the fermion spin and the angle measured by the S.\,Major operator is the same if we take the limit of large \(j_i\) in both cases.
\subsection{Four-valent LQG setup}
\subsubsection{Exact solution}
Analogous to the three-valent case we now let all \(j_i\) be non vanishing. We couple \(j_2\), \(j_4\) and \(j_1\), \(j_3\) to the internal edges \(j_{24}\) and \(j_{13}\) respectively. We fix \(j_{24}\) such that we can write \(j_{13}=j_{24}+\frac{r}{2}\), \(r=\pm1\):
\begin{equation}
    \ket{r}^{(j_{24})}_0=\qty(-1)^{2j_{24}}
    \begin{tikzpicture}[baseline=($(j1.base)!.5!(j2.base)$)]
    \begin{feynman}[inline=($(j1.base)!.5!(j2.base)$)]
        \vertex [small, dot] (int_f) {}; 
        \vertex [above right=of int_f,dot] (int_13) {}; 
        \vertex [below left=of int_f,dot] (int_24) {}; 
        \vertex [above=of int_13] (j1) {\(j_1\)};
        \vertex [right=of int_13] (j3) {\(j_3\)}; 
        \vertex [below=of int_24] (j2) {\(j_2\)}; 
        \vertex [left=of int_24] (j4) {\(j_4\)}; 
        \vertex [below right=of int_f] (f) {\(\frac{1}{2}\)};
        \diagram*{
            (int_f) -- [fermion, edge label=\(j_{24}+\frac{r}{2}\)] (int_13),
            (int_13) -- [fermion] (j1),
            (int_13) -- [fermion] (j3),
            (int_f) -- [fermion, edge label=\(j_{24}\)] (int_24),
            (int_24) -- [fermion] (j2),
            (int_24) -- [fermion] (j4),
            (int_f) -- [charged scalar] (f),
        };
    \end{feynman}
    \end{tikzpicture}.
\end{equation}

Essentially the same argument as in \cref{sec:3valent setup} is valid with the indices of the symmetric operator \(O\) now becoming tuples, e.g. \((r,j_{14})\). After recoupling to a basis \(\ket{r}^{(j_{14})}_1\) where \(j_1\), \(j_4\) and \(j_2\), \(j_3\) couple to \(j_{14}\) and \(j_{23}\) respectively we arrive at
\begin{equation}
    \cos(\alpha) = (2j_{24}+1)\qty(2\qty(j_{24}+\frac{1}{2})+1)\sum_{\substack{j_{14},r\\j_{23}=j_{14}+r/2}}(2j_{14}+1)(2j_{23}+1)\,r\,
    {\begin{Bmatrix}
        j_1 & j_3 & j_{24}+\frac{1}{2} \\
        j_4 & j_2 & j_{24} \\
        j_{14} & j_{23} & \frac{1}{2}
    \end{Bmatrix}}^2 \,.\label{eq:pm_O_pm} 
\end{equation}
Analogously to before we can interpret \(\alpha+\pi\) as the angle between the surface normals \(\vec{n}_{0,1}\) associated to \(S_{0,1}\) with respect to the fermion spin. However, in this case a direct computation of \(\cos(\alpha)\) seams not possible due to the additional degree of freedom \(j_{24}\).
\subsubsection{Interpretation of the angle \(\alpha\)}
Although we were not able to calculate \(\cos(\alpha)\) explicitly, we strongly assume that the angle will be the same as the one we get when measuring the angle between the surface normals with the S.\,Major angle operator.

As in the three valent case we calculate the matrix element
\begin{align}\label{eq:fourvalent Major p_cos_p}
    ^{(j_{24})}_{\hphantom{j_{24})}0}\kern-2pt\mel{\pm}{\cos_{S_0S_1,v}}{\pm}^{(j_{24})}_0 &\approx \sum_{j_{14}} p_{\pm}(j_{14})\frac{(J_1)^2+(J_2)^2-(J_3)^2-(J_4)^2}{2J_{24}J_{14}}, \\
    ^{(j_{24})}_{\hphantom{j_{24})}0}\kern-2pt\mel{\pm}{\cos_{S_0S_1,v}}{\mp}^{(j_{24})}_0&\approx 0
\end{align}
where
\begin{equation}
    p_{\pm}(j_{14})=\sum_{j_{23}}(2j_{24}+1)\qty(2\qty(j_{24}\pm\frac{1}{2})+1)(2j_{14}+1)(2j_{23}+1)\ninejsymbol{j_1}{j_3}{j_{24}\pm\frac{1}{2}}{j_4}{j_2}{j_{24}}{j_{14}}{j_{23}}{\frac{1}{2}}^2 \,.
\end{equation}
We recover the same fraction in the sum if we calculate the matrix element without introducing the fermion edge in the first place and consider only the \(j_i\), given \(J_{24}\gg\frac{1}{2}\) (and similar for \(J_{14}\)):
\begin{equation}
    \sum_{j_{14}} p(j_{14})\frac{(J_1)^2+(J_2)^2-(J_3)^2-(J_4)^2}{2J_{24}J_{14}} \quad,\quad p(j_{14})=(2j_{24}+1)(2j_{14}+1)\sixjsymbol{j_1}{j_3}{j_{24}}{j_2}{j_4}{j_{14}}^2 \,.
\end{equation}
As expected, we could have also arrived here from \cref{eq:fourvalent Major p_cos_p} by replacing \(\frac{1}{2}\rightarrow0\). The fraction above turns out to be an angle in a tetrahedron formed by the six vectors \(\vec{J}_i\), \(\norm\big{\vec{J}_i}=J_i\) where \(\vec{J}_{24}=\vec{J}_2+\vec{J}_4=-\vec{J}_1-\vec{J}_3\) and \(\vec{J}_{14}=\vec{J}_1+\vec{J}_4=-\vec{J}_2-\vec{J}_3\). Namely 
\begin{equation}
    \cos(\alpha)=\frac{(-\vec{J}_{24})\cdot\vec{J}_{14}}{\norm\big{-\vec{J}_{24}}\norm\big{\vec{J}_{14}}}=\frac{{J_1}^2+{J_2}^2-{J_3}^2-{J_4}^2}{2J_{24}J_{14}} \,,
\end{equation}
for a visualization see \cref{fig: tetrahedron}. As \(p(j_{14})>0\) and \(\sum_{j_{14}}p(j_{14})=1\), \cref{eq:fourvalent Major p_cos_p} states that the eigenvalues of the cosine operator are given by a classical expectation value of angles \(\cos(\alpha)\) in tetrahedra with one side length varying, the probability distribution is given by \(p(j_{14})\). This \enquote{quantum fuzziness} in the tetrahedron and the angle respectively is somewhat expected as a classical tetrahedron is uniquely fixed by six lengths, however choosing the intertwiner with \(j_{1,2,3,4,24}\) fixed only fixes five \(J_i\)'s.
\pagebreak
\section{Summary and discussion}\label{sec:summary}%
%We have reviewed how fermions are introduced into the kinematical framework of LQG, namely by adding a fermionic Fock space to the gravitational Ashtekar Lewandowski Hilbert space. States thereon are comprised of spin network states and \(n\) particle states which are created via creation operators \(c^\dagger\) acting on the vacuum. We have also referred the reader to \cite{Mansuroglu:2021azm} where a spin operator is introduced based on the quantization of the classical spin density in terms of \(\sigma_i\) and \(c,c^\dagger\).

We started our considerations of fermionic correlations in LQG by demonstrating that the notion of fermionic entanglement in loop quantum gravity is subtle. In particular, we showed that some obvious ways to investigate it fail.

Based on the idea of the spin projection in a given direction in space in quantum mechanics, we then constructed an operator that corresponds to the spin component in the direction normal to a given surface. This is an operator that contains, besides the premium spin operator, components of the quantum gravitational field corresponding to the spatial geometry at the surface.
On a technical level, this operator could be constructed using the standard gravitational electric flux operator, with the fermionic spin operator 
 \(S^i(x)\) replacing the classical smearing function.
 
After normalization and symmetrization we ended up with the cosine operator \(\cos_{\mathcal{F}S}\) measuring the angle of between spin and surface normal.
%By carrying over the action of all operators involved to the intertwiners of the spin network states \(\widehat{\cos}_{\mathcal{F}S}\) acts on, 
We were able to calculate the spectrum of this operator for an arbitrary surface \(\mathcal{F}\). It turns out that the spectrum is bounded by the quantum mechanical values \(\pm\frac{1}{\sqrt{3}}\) corresponding to a spin component of $\vec{S}\cdot \vec{n} = \pm \frac{1}{2}$. In the limit of large spins, those values are obtained. The limit of large spins can be interpreted very loosely as a regime in which the spatial geometry has small quantum fluctuations. The additional eigenvalues, as compared to the quantum mechanical case, can thus be interpreted as due to the quantum nature of geometry, and in particular, of the direction that the spin gets projected into. The sign $\operatorname{sgn}_{\mathcal{F}S}$ gives a convenient observable to use to construct the CHSH observable. In the regime of large spins, it coincides with  $2\vec{S}\cdot \vec{n}$. 

%In the case of vanishing tangential edges the cosine operator, to a good approximation, acts only by multiplication with a sign which lead us to define the operator \(\widehat{\operatorname{sgn}}_{\mathcal{F}S}\).

Using $\operatorname{sgn}_{\mathcal{F}S}$ we were able to construct a Bell test for loop quantum gravity: Alice and Bob each
%each define a surface, dividing the region around their respective vertices into four sections \(I\) to \(IV\). 
%They 
measure the sign of the spin projection normal to two surfaces of their choosing and consider correlations of the measurement. We showed that there exist states whose correlations maximally violate the Bell-CHSC inequality. 

The degree of correlation in a given state can be  described by two angles \(\alpha_{A,B}\). They are the analogs of the angles between the two directions in which
Alice and Bob each measure the spin components in the quantum mechanical setup. Here one can interpret them as the respective angles between the two surface normals (up to a shift by \(\pi\)). This interpretation is supported by the fact that measuring the angle between the surface normals using the operator of \cite{Major:1999mc} without referring to the fermion spin yields the same result in cases where we can calculate them explicitly, and there are strong hints that this will hold in general. 

Curiously, the angles \(\alpha_{A,B}\) also admit another geometric interpretation. In a particular case (no edges in wedge \(IV\) created by the two surfaces) the angle $\alpha$ can be calculated explicitly. It turns out 
that it is the external angle in a triangle formed by the recoupling spins \(j_i\) of the three wedges \(I\), \(II\), \(III\) at the respective vertex. In the general case, we were not able to calculate \(\alpha_{A,B}\) explicitly, but comparing to \cite{Major:1999mc} leads us to expect that in this case \(\alpha_{A,B}\) are respectively expectation values of angles between opposite edges in a tetrahedron spanned by recoupling spins \(j_i\). The fact that the nature of the angle is now probabilistic is due to a free parameter manifesting itself as an undetermined side length of the tetrahedron.

There are several steps on can take from here on. First, one might further investigate the spectrum of \(\cos_{\mathcal{F}S}\) . As it contains countably infinitely many points it might be interesting to see whether there are accumulations points next to the quantum mechanical values. Perhaps there may also be an experimental setup in which the deviations of the behaviour of the spin components $\vec{S}\cdot \vec{n}{}_\mathcal{F}$ from the quantum mechanical case could be detectable.

%Also one might be able to come up with experiments which show a different phenomenological behaviour of the newly introduced \(\widehat{\cos}_{\mathcal{F}S}\) compared to the standard quantum mechanical one.

Another idea that we are currently pursuing \cite{Sahlmann_Zeiss_TBD} is to investigate amplitudes for the generation of such Bell states, or at least states with non-classical correlations, via spinfoam models. 
%This would in principle allow to study the dynamics of matter entanglement generation in loop quantum gravity. 
If one could show that there exist foams which introduce fermionic entanglement as measured by the \(\cos_{\mathcal{F}S}\) (or \(\operatorname{sgn}_{\mathcal{F}S}\)) operator, this would demonstrate a first principles calculation of matter entanglement generation through quantum gravity in loop quantum gravity. A refinement of these ideas may lead to predictions of effects observable in proposed experiments \cite{Bose:2017nin,Marletto:2017kzi}.  

\acknowledgments
HS acknowledges the contribution from COST Action BridgeQG (CA23130), supported by COST (European Cooperation in Science and Technology). HS thanks Carlo Rovelli for pointing out to him that consideration of some kind of relative observable involving fermion spin and some sort of reference frame might be useful. The spin components normal to a surface are a direct implementation of this idea. HS would also like to thank Lennard Kossmann and Johannes Gro{\ss}e for various discussions on fermion correlations in loop quantum gravity. 

% Bibliography

%% [A] Recommended: using JHEP.bst file
%% \bibliographystyle{JHEP}
%% \bibliography{biblio.bib}

%% or
%% [B] Manual formatting (see below)
%% (i) We suggest to always provide author, title and journal data or doi:
%% in short all the informations that clearly identify a document.
%% (ii) please avoid comments such as "For a review'', "For some examples",
%% "and references therein" or move them in the text. In general, please leave only references in the bibliography and move all
%% accessory text in footnotes.
%% (iii) Also, please have only one work for each \bibitem.

%%%autocite style!!!!
\newpage

\bibliographystyle{JHEP}
\bibliography{biblio}

\end{document}